\documentclass[11pt,reqno]{amsart}
\usepackage{amsmath,amssymb,amsthm,upref,graphicx,mathrsfs,enumerate}
\usepackage{physics}
\usepackage{textcomp} 
\usepackage{array} 
\usepackage{textcomp} 
\usepackage{enumitem}

\usepackage[colorlinks=true, allcolors=blue]{hyperref}

\usepackage{color,soul}

\usepackage{amsthm}

\usepackage{xspace}

\newcommand*{\doi}[1]{\href{https://dx.doi.org/#1}{doi: #1}}

\numberwithin{equation}{section}
\allowdisplaybreaks

\newtheorem{theorem}{Theorem}[section]
\newtheorem{prop}[theorem]{Proposition}
\newtheorem{lemma}[theorem]{Lemma}
\newtheorem{cor}[theorem]{Corollary}

\theoremstyle{definition}
\newtheorem{definition}[theorem]{Definition}
\newtheorem{example}[theorem]{Example}
\newtheorem{remark}[theorem]{Remark}

\newcommand{\R}{\mathbb{R}}
\newcommand{\C}{\mathbb{C}}
\newcommand{\D}{\mathbb{D}}

\newcommand{\cX}{\mathcal{X}}

\newcommand{\cY}{\mathcal{Y}}

\newcommand{\cH}{\mathcal{H}}
\newcommand{\cK}{\mathcal{K}}
\newcommand{\cC}{\mathcal{C}}
\newcommand{\cR}{\mathcal{R}}

\newcommand{\fS}{\mathfrak{S}}

\newcommand{\supp}{{\rm supp}\,}

\def\coh{\rm{coh}}

\DeclareMathOperator*{\esssup}{ess\,sup}
\DeclareMathOperator*{\spn}{span}

\makeatletter
\def\@tocline#1#2#3#4#5#6#7{\relax
	\ifnum #1>\c@tocdepth 
	\else
	\par \addpenalty\@secpenalty\addvspace{#2}%
	\begingroup \hyphenpenalty\@M
	\@ifempty{#4}{%
		\@tempdima\csname r@tocindent\number#1\endcsname\relax
	}{%
		\@tempdima#4\relax
	}%
	\parindent\z@ \leftskip#3\relax \advance\leftskip\@tempdima\relax
	\rightskip\@pnumwidth plus4em \parfillskip-\@pnumwidth
	#5\leavevmode\hskip-\@tempdima
	\ifcase #1
	\or\or \hskip 1em \or \hskip 2em \else \hskip 3em \fi%
	#6\nobreak\relax
	\hfill\hbox to\@pnumwidth{\@tocpagenum{#7}}\par
	\nobreak
	\endgroup
	\fi}
\makeatother

\begin{document}

	\title[Operator-valued Schatten spaces and Quantum entropies]{Operator-valued Schatten spaces and Quantum entropies}
	
	\author[S. Beigi \& M. M. Goodarzi]{Salman Beigi \& Milad M. Goodarzi}

	\address[Salman Beigi]{School of Mathematics, Institute for Research in Fundamental
		Sciences (IPM), P.O. Box 19395-5746, Tehran, Iran}
	\email{salman.beigi@gmail.com}
	
	\address[Milad M. Goodarzi]{School of Mathematics, Institute for Research in Fundamental
		Sciences (IPM), P.O. Box 19395-5746, Tehran, Iran}
	\email{milad.moazami@gmail.com}

	\subjclass[2010]{}

	\begin{abstract}
		Operator-valued Schatten spaces were introduced by G.\ Pisier as a noncommutative counterpart of vector-valued $\ell_p$-spaces. This family of operator spaces forms an interpolation scale which makes it a powerful and convenient tool in a variety of applications. In particular, as the norms coming from this family naturally appear in the definition of certain entropic quantities in Quantum Information Theory (QIT), one may apply Pisier's theory to establish some features of those quantities. Nevertheless, it could be quite challenging to follow the proofs of the main results of this theory from the existing literature. In this article, we attempt to fill this gap by presenting the underlying concepts and ideas of Pisier's theory in a self-contained way which we hope to be more accessible, especially for the QIT community at large. Furthermore, we describe some applications of this theory in QIT. In particular, we prove a new uniform continuity bound for the quantum conditional R\'enyi entropy.
	\end{abstract}

	\maketitle

	\tableofcontents

	\section{Introduction}
	
	In recent years, tools from noncommutative analysis have proven fruitful in understanding the primitives of Quantum Information Theory (QIT), where--mostly due to the noncommutative nature of the theory--even finding viable quantum analogs of certain information-theoretic quantities turns out to be nontrivial, and various challenges arise to be overcome (see, e.g.,~\cite{DJKR, DB14, WWY14, BDR20, BJLRF21}). Analytical methods have been exploited in classical information theory just as well, where quantum phenomena are nonexistent (see, e.g.,~\cite{RaginskySason, LCCV19, KamathAnantharam, LvHV17}).

	R\'enyi divergence (or R\'enyi relative entropy) is a useful concept in information theory \cite{VanErfenHarremoes}.
	A quantum analog of this entropic measure is tricky to define; since arbitrary density operators $\rho, \sigma$ do not necessarily commute, there are multiple candidates for a definition of quantum R\'enyi relative entropy. One way is to define it as the so-called \emph{sandwiched R\'enyi divergence} \cite{WWY14, MDSFT13} given by
	$$
	D_\alpha(\rho\| \sigma) = \frac{\alpha}{\alpha-1}\log \Big\|  \sigma^{\frac{1-\alpha}{2\alpha}}\rho\sigma^{\frac{1-\alpha}{2\alpha}}\Big\|_\alpha,
	$$
	where $\|\cdot\|_\alpha$ denotes the $\alpha$-Schatten norm.
	For applicability and operational significance \cite{WWY14, MosonyiOgawa}, this quantity has to satisfy certain properties (e.g., data-processing inequality)~\cite{MDSFT13,  FrankLieb, Beigi13}. Since the Schatten norms form an interpolation scale, the above expression of sandwiched R\'enyi divergence allows one to utilize the complex interpolation method in establishing the aforementioned properties~\cite{Beigi13}.

	Consider a composite quantum system with bipartite Hilbert space $\cH_X\otimes \cH_Y$. In analogy with the classical case, a quantum generalization of the conditional R\'enyi entropy can be defined, for a density operator $\rho_{XY}$ on $\cH_X\otimes \cH_Y$, using the (sandwiched) R\'enyi divergence \cite{MDSFT13}:
	$$
	H_\alpha(X|Y)_\rho := - \inf D_\alpha(\rho_{XY}\| I_X\otimes \sigma_Y),
	$$
	where the infimum runs over all density operators $\sigma_Y$ acting on $\cH_Y$. By expressing this quantity in terms of the Schatten norms as above, we observe that the norms on the \emph{operator-valued Schatten spaces}\footnote{These are the discrete case of non-commutative vector valued $L_p$-spaces of \cite{Pisier2}.} appear in the guise of the divergence (see below). These spaces were first introduced by Pisier~\cite{Pisier2} as a noncommutative counterpart to the vector-valued space $\ell_p(\mathcal{X})$ of $p$-summable sequences in a Banach space $\mathcal{X}$ for $1\leq p\leq\infty$. Roughly speaking, Pisier's theory interweaves the interpolating structure of Schatten classes $\mathcal{S}_p$ with an operator space structure\footnote{This operator space structure is essential for the theory to have all the desired properties, such as duality.} on $\mathcal{X}$ to construct a whole family of operator spaces, denoted by $\mathcal{S}_p[\mathcal{X}]$.\footnote{The precise definition of $\mathcal{S}_p[\mathcal{X}]$ is given in Section \ref{sec:NCVV}.}

	For our purposes in this work, the most relevant norms are obtained when $\mathcal{X}$ is taken to be the Schatten class operator space $\mathcal{S}_q$, $1\leq q \leq \infty$. In this way we obtain a two-parameter family of spaces $\mathcal{S}_p[\mathcal{S}_q]$, $1\leq p, q\leq \infty$, that form an interpolation scale and whose elements act on a tensor product Hilbert space, say $\cH_X\otimes \cH_Y$. Assuming $p\leq q$, the norm of an operator $M_{XY}$ in $\mathcal{S}_p[\mathcal{S}_q]$ is quantitatively characterized by the variational expression~\cite{Pisier2,Junge96}
	\begin{align}\label{eq:p-q-norm-intro}
		\|M_{XY}\|_{\mathcal{S}_p[\mathcal{S}_q]}=\|M_{XY}\|_{(p, q)} := \inf \|A_X\|_{\mathcal{S}_{2r}}\|T_{XY}\|_{\mathcal{S}_q}\|B_Y\|_{\mathcal{S}_{2r}},
	\end{align}
	where the infimum is taken over all factorizations $M_{XY}= (A_X\otimes I_Y)T_{XY}(B_X\otimes I_Y)$, and $r$ is defined by
	$$
	\frac{1}{p}=\frac 1q+\frac 1r.
	$$
	The variational formula \eqref{eq:p-q-norm-intro} is the subtle bridge that links the norm on $\mathcal{S}_p[\mathcal{S}_q]$ with quantum entropies as it allows to formulate the quantum conditional R\'enyi entropy in terms of the $(1, \alpha)$-norm:
	\begin{align*}
		H_\alpha(X|Y)_\rho = \frac{\alpha}{1-\alpha} \log \big\|\rho_{YX}\big\|_{(1, \alpha)}.
	\end{align*}
	This new formulation makes it possible, as we will see, to utilize the power of methods and techniques from interpolation theory and operator space theory to investigate characteristics of these entropic measures. From a broader perspective, this connection suggests potential applications of Pisier's theory in other areas of QIT as well. 
	
	Based on the above observation, we believe that it is worth the effort to make the theory of operator-valued Schatten spaces more accessible, particularly to the QIT community at large; this is our primary goal in this paper. To this end, we describe all the necessary constituent components of the theory and provide detailed proofs of the major results in a self-contained way, whereas the unnecessary parts (with respect to QIT) are dismissed. Although for the most parts we do not confine ourselves to finite-dimensional spaces, we will not refrain from doing so when it contributes to simplicity of presentation and clarity of the main ideas. To showcase advantages of this theory, we will also illustrate how it facilitates the study of quantum R\'enyi quantities.
	
	We remark that following Pisier~\cite{Pisier2} we employ the theory of operator spaces to define operator-valued Schatten spaces and prove their properties. An alternative approach for defining these spaces would be taking the concrete formula~\eqref{eq:p-q-norm-intro} as the starting point, without appealing to operator spaces. This approach is already taken in~\cite{JungeParcet},\footnote{This approach was also raised by Marco Tomamichel and Alexander McKinlay in a discussion with the first author.} yet it would not necessarily reduce the burden since, for instance, proving from scratch that $\|\cdot\|_{(p, q)}$ satisfies the triangle inequality is not straightforward---let alone proving the fact that these norms form an interpolation scale (see \cite{JungeParcet} for details). Moreover, so as to be able to utilize the full power of operator-valued Schatten spaces in applications, we need the theory of operator spaces.

	The rest of this paper is organized as follows. In Section~\ref{sec:interpolation} we present a concise description of the complex interpolation method due to Calder\'{o}n \cite{Cald}, and point out its basic properties. Next, in Section~\ref{sec:operator-spaces} we review the theory of operator spaces and define the key concepts that will be used in the sequel. More specifically, we first introduce the dual, column and row operator spaces, and explain why the latter two are not isomorphic as operator spaces. We then introduce Pisier's interpolation in the category of operator spaces, from which we derive the natural operator space structure of Schatten classes. Furthermore, we introduce the Haagerup tensor product of operator spaces and conclude the section with the crucial result that the Haagerup tensor product ``commutes" with interpolation. In Section~\ref{sec:NCVV} we finally proceed to give the definition of operator-valued Schatten spaces as a certain Haagerup tensor product. We will show that they form an interpolation scale and prove the aforementioned variational characterization for their norms. Section~\ref{sec:applications} is devoted to applications of these spaces in QIT. Here, we provide short proofs of some existing results, and prove some new ones based on the theory developed in the previous sections. In particular, we establish a new uniform continuity bound for the quantum conditional R\'enyi entropy.
	
	We should finally mention that for the convenience of a reader who might be interested only in applications of the above-mentioned material, we have strived to include a minimal number of proofs in the main body of the paper. Nevertheless, to make it self-contained, several appendices follow the main sections to provide some of the proofs.

	\section{Interpolation theory: The complex method}\label{sec:interpolation}
		Recently, interpolation theory has proven useful in establishing properties of quantum entropic quantities and studying properties of quantum channels. 	In this section, we present a brief overview of the complex interpolation method due to Calder\'{o}n~\cite{Cald}. We refer the reader to ~\cite{BerghLofstrom, Lunardi} for more details on interpolation theory.
		
	Banach spaces in this paper are all considered to be complex.
	For Banach spaces $(\cX,\|\cdot\|_{\mathcal{X}})$ and $(\cY,\|\cdot\|_{\mathcal{Y}})$, 
	unless otherwise explicitly stated, the equality $\mathcal{X}=\mathcal{Y}$ means that $\mathcal{X},\mathcal{Y}$ are the same spaces (up to linear isomorphism) with \emph{equivalent norms} in the sense that there are absolute constants $c,C>0$ such that $c\|\cdot\|_{\mathcal{X}}\leq \|\cdot\|_{\mathcal{Y}}\leq C\|\cdot\|_{\mathcal{X}}$.
	 
	We denote by $B(\mathcal{X},\mathcal{Y})$ the Banach space of all bounded linear operators from $\mathcal{X}$ to $\mathcal{Y}$, equipped with the usual operator norm
	$$
	\|T\|_{B(\mathcal{X},\mathcal{Y})}:=\sup\big\{\|Tx\|_{\mathcal{Y}}:x\in {X},\|x\|_{\mathcal{X}}\leq 1\big\}.
	$$
	When $\cX=\cY$, let $B(\cX)=B(\cX, \cY)$. We denote the dual of $\cX$ by $\cX^*= B(\cX, \C)$.

	We write $\mathcal{X}\hookrightarrow \mathcal{Y}$ if there is a continuous embedding from $\mathcal{X}$ into $\mathcal{Y}$.
	A couple of Banach spaces $(\mathcal{X},\mathcal{Y})$ is said to be \emph{compatible} if there is a Hausdorff topological vector space $\mathcal{V}$ such that $\mathcal{X},\mathcal{Y}\hookrightarrow \mathcal{V}$. 
	For such a compatible couple we may  consider $\mathcal{X}+\mathcal{Y}$ as a linear subspace of $\mathcal{V}$, and for any $v\in \cX+\cY$ define
	$$
	\|v\|_{\mathcal{X}+\mathcal{Y}}:=\inf\{\|x\|_{\mathcal{X}}+\|y\|_{\mathcal{Y}}:x\in\mathcal{X},y\in\mathcal{Y},v=x+y\}.
	$$
	This norm turns $\mathcal{X}+\mathcal{Y}$ into a Banach space.

	In order to prepare the setup, consider the strip 
	$$\mathbb{S}:=\{z\in\mathbb{C}:0\leq\text{Re}(z)\leq1\},$$ 
	in the complex plane, and let $\mathcal{X}$ be a Banach space. A function $f:\mathbb{S}\rightarrow \mathcal{X}$ is said to be \emph{holomorphic} in the interior of $\mathbb{S}$, if for every $x^*\in \mathcal{X}^*$ the complex function $z\mapsto \langle x^*, f(z)\rangle  $ is holomorphic in the interior of $\mathbb{S}$. It is not hard to verify that the \emph{maximum principle} holds for such holomorphic functions. Namely,  for a bounded continuous function $f:\mathbb{S}\rightarrow \mathcal{X}$ that is holomorphic in the interior of $\mathbb{S}$, we have
	\begin{equation*}
		\sup_{z\in S}\|f(z)\|_{\mathcal{X}}\leq\max\big\{\sup_{t\in\mathbb{R}}\|f(it)\|_{\mathcal{X}},\sup_{t\in\mathbb{R}}\|f(1+it)\|_{\mathcal{X}}\big\}.
	\end{equation*}

	Let $(\mathcal{X},\mathcal{Y})$ be a compatible couple of complex Banach spaces. We denote by $\mathscr{F}(\mathcal{X},\mathcal{Y})$ the space of all bounded continuous functions $f:\mathbb{S}\to \mathcal{X}+\mathcal{Y}$ that are holomorphic in the interior of $\mathbb{S}$ such that
	$$
	t\mapsto f(it)\in \mathcal{X} \ \ \ \text{ and } \ \ \ t\mapsto f(1+it)\in \mathcal{Y}
	$$
	are continuous functions.  We equip $\mathscr{F}(\mathcal{X},\mathcal{Y})$ with the norm
	$$
	\|f\|_{\mathscr{F}(\mathcal{X},\mathcal{Y})}:=\max\big\{\sup_{t\in\mathbb{R}}\|f(it)\|_{\mathcal{X}},\sup_{t\in\mathbb{R}}\|f(1+it)\|_{\mathcal{Y}}\big\}.
	$$
	Using the maximum principle mentioned above, it is not difficult to verify that $\mathscr{F}(\mathcal{X},\mathcal{Y})$ 
	is a Banach space.
	
	We are now ready to define the complex interpolation spaces $[\mathcal{X},\mathcal{Y}]_{\theta}$, $0\leq \theta\leq 1$.
	
	\begin{definition}
		Let $(\cX, \cY)$ be a compatible couple of Banach spaces. 
		For every $0<\theta<1$, we define $[\mathcal{X},\mathcal{Y}]_{\theta}$ to be the linear space
		$$
		[\mathcal{X},\mathcal{Y}]_{\theta}:=\big\{f(\theta):f\in\mathscr{F}(\mathcal{X},\mathcal{Y})\big\},
		$$
		endowed with the norm
		$$
		\|v\|_{[\mathcal{X},\mathcal{Y}]_{\theta}}:=\inf\big\{\|f\|_{\mathscr{F}(\mathcal{X},\mathcal{Y})}:f(\theta)=v, f\in\mathscr{F}(\mathcal{X},\mathcal{Y})\big\}.
		$$
		It can be verified that $[\mathcal{X},\mathcal{Y}]_{\theta}$ is a Banach space. 
		
	\end{definition}

	From the definition it is clear that $[\mathcal{X},\mathcal{X}]_{\theta}=\mathcal{X}$ with identical norms. Moreover, we have (\cite[Theorem 4.2.1]{BerghLofstrom})
	$$	[\mathcal{X},\mathcal{Y}]_{\theta}=[\mathcal{Y},\mathcal{X}]_{1-\theta}.$$
	\emph{Reiteration theorem} is another interesting property of interpolation spaces. Let $(\cX_0, \cX_1)$ be a compatible couple and let $\cX_\theta=[\cX_0, \cX_1]_\theta$, $0<\theta<1$. Then, a special case of the reiteration theorem (\cite[Theorem 4.6.1]{BerghLofstrom}) says that if $\cX_0\subseteq \cX_1$ as vector spaces, for $0\leq \theta_0, \theta_1\leq 1$ we have
	\begin{equation}\label{eq:reiteration}
		[\mathcal{X}_{\theta_0},\mathcal{X}_{\theta_1}]_\theta=\mathcal{X}_{(1-\theta)\theta_0+\theta\theta_1}, \qquad 0<\theta<1.
	\end{equation}
	Interpolation spaces behave nicely under taking the Banach dual as well. Suppose that one of the spaces $\cX_0, \cX_1$ of a compatible couple $(\cX_0, \cX_1)$ are finite-dimensional (or more generally, reflexive). Then, letting $\cX_\theta=[\cX_0, \cX_1]_\theta$, $0<\theta<1$, we have
	$$\cX_\theta^* = [\cX^*_0, \cX^*_1]_\theta,$$
	where $\cX^*_\theta$ is the Banach dual of $\cX_\theta$ (see \cite[Corollary 4.5.2]{BerghLofstrom}).

	\begin{remark}\label{rem:interrpolation-scaling}
		For every $v\in [\cX, \cY]_\theta$ we have
		\begin{align*}
			\|v\|_{[\mathcal{X},\mathcal{Y}]_{\theta}}=\inf_{\stackrel{f\in \mathscr F(\cX, \cY)}{f(\theta)=v}} \Big( \sup_t \|f(it)\|_{\cX}\Big)^{1-\theta} \Big(\sup_t \|f(1+it)\|_{\cY}   \Big)^\theta.
		\end{align*}
		To see this, we note that by definition
		$$	\|v\|_{[\mathcal{X},\mathcal{Y}]_{\theta}}:=\inf_{\stackrel{f\in \mathscr F(\cX, \cY)}{f(\theta)=v}}\max \big\{  \sup_t \|f(it)\|_{\cX}, \sup_t \|f(1+it)\|_{\cY}   \big\},$$
		and 
		$$ \Big( \sup_t \|f(it)\|_{\cX}\Big)^{1-\theta} \Big(\sup_t \|f(1+it)\|_{\cY}   \Big)^\theta\leq \max \big\{  \sup_t \|f(it)\|_{\cX}, \sup_t \|f(1+it)\|_{\cY}   \big\}.$$
		On the other hand, given $f\in \mathscr F(\cX, \cY)$ with $f(\theta)=v$, the map $f_\alpha(z) = e^{\alpha(z-\theta)}f(z)$, for arbitrary $\alpha\in \mathbb R$, also belongs to $\mathscr F(\cX, \cY)$ and satisfies $f_\alpha(\theta)=v$. Moreover, for every $t\in \R$ we have $\|f_\alpha(it)\|_\cX = e^{-\alpha\theta}\|f(it)\|_{\cX}$ and $\|f_\alpha(1+it)\|_\cY = e^{\alpha(1-\theta)}\|f(1+it)\|_{\cY}$. 
		Then, optimizing over the choice of $\alpha\in \R$, we obtain
		$$\inf_{\alpha\in \R} \max \big\{  \sup_t \|f_\alpha(it)\|_{\cX}, \sup_t \|f_\alpha(1+it)\|_{\cY}   \big\}=  \Big( \sup_t \|f(it)\|_{\cX}\Big)^{1-\theta} \Big(\sup_t \|f(1+it)\|_{\cY}   \Big)^\theta.$$
		Putting these together, the desired equality follows.
		\hfill $\Box$ 
	\end{remark}

	\begin{remark}\label{rem:interpolation-ineq}
		Considering the constant map $f(z)= v$, Remark~\ref{rem:interrpolation-scaling} shows that 
		\begin{align*}\label{eq:interrpolation-ineq}	
			\|v\|_{[\mathcal{X},\mathcal{Y}]_{\theta}}\leq \|v\|_{\cX}^{1-\theta} \|v\|_{\cY}^\theta, \ \ \ \text{for} \ v\in \mathcal{X}\cap \mathcal{Y}.
		\end{align*}
	\end{remark}

	\medskip
	\begin{example}
		$L_p$-spaces are important examples of complex interpolation. Let $(\Omega,\mu)$ be a $\sigma$-finite measure space, and for $1\leq p<\infty $ let $L_p(\Omega)=L_p(\Omega, \C)$ be the space of complex measurable functions $h$ on $\Omega$ with $\int |h|^p\dd \mu<+\infty$, equipped with norm $\|h\|_p =\big(\int |h|^p\dd \mu\big)^{1/p}$. For $p=+\infty$, we define $L_{\infty}(\Omega,\C)$ with norm $\|h\|_{\infty} =\esssup |h|$.\footnote{$\esssup |h|:=\inf\{a\in\mathbb{R}:|h(z)|\leq a \ \text{for almost all} \ z\in\Omega\}$.} Then, $L_p(\Omega)$ for any $1\leq p\leq +\infty$ is a Banach space and for $1\leq p_0,p_1\leq\infty$, $0<\theta<1$ we have
		\begin{equation}\label{eq:interpolation-Lp}
			[L_{p_0}(\Omega),L_{p_1}(\Omega)]_{\theta}=L_{p_{\theta}}(\Omega),
		\end{equation}
		with identical norms, where\footnote{As a convention we set $\frac1{\infty}=0$.}
		$$\frac1{p_{\theta}}=\frac{1-\theta}{p_0}+\frac{\theta}{p_1}.$$ 
	\end{example}

	The following theorem states an important bound on the operator norm of bounded maps between interpolation spaces. 
	
	\begin{theorem}[\cite{Lunardi}]
		Let $(\mathcal{X}_0,\mathcal{X}_1)$ and $(\mathcal{Y}_0,\mathcal{Y}_1)$ be compatible couples, and $0<\theta<1$. If $T:\mathcal{X}_0+\mathcal{X}_1\rightarrow \mathcal{Y}_0+\mathcal{Y}_1$ is a linear operator such that $T|_{\mathcal{X}_j}\in B(\mathcal{X}_j,\mathcal{Y}_j)$, $j=0,1$, then $T|_{[\mathcal{X}_0,\mathcal{X}_1]_{\theta}}\in B([\mathcal{X}_0,\mathcal{X}_1]_{\theta},[\mathcal{Y}_0,\mathcal{Y}_1]_{\theta})$, and
		\begin{equation*}
			\|T\|_{B([\mathcal{X}_0,\mathcal{X}_1]_{\theta},[\mathcal{Y}_0,\mathcal{Y}_1]_{\theta})}\leq\|T\|_{B(\mathcal{X}_0,\mathcal{Y}_0)}^{1-\theta}\|T\|_{B(\mathcal{X}_1,\mathcal{Y}_1)}^{\theta}.
		\end{equation*}
		\label{thm:interpolation-bound}
	\end{theorem}

	
	\section{Operator spaces}\label{sec:operator-spaces}
	The space of bounded operators $B(\cH)$ acting on a Hilbert space $\cH$ naturally appears in the study of a quantum system described by $\cH$. If the quantum system interacts with another quantum system described by an $n$-dimensional Hilbert space, then we need to study operators acting on the tensor product Hilbert space. Such an operator can be thought of as an $n\times n$ matrix whose entries belong to $B(\cH)$. From this point of view, not only $B(\cH)$, but also matrices with entries in $B(\cH)$ become relevant in quantum information theory.  Operator space theory is a theoretical ground for the study of such objects. 	
	In this section we give a brief overview of operator space theory and refer to~\cite{Pisier3} for more details.
	
	We denote by $\ell_2^n$ the Hilbert space of $n$-tuples of complex numbers, and let $\ell_2=\ell_2^{\infty}$. For a Hilbert space $\cH$ and $n\in \mathbb N$, we let $\cH^n  = \oplus_{i=1}^n \cH=\ell_2^n\otimes_2 \cH$ where $\otimes_2$ denotes the \emph{Hilbertian tensor product}. The space of all $m\times n$ matrices with entries in a linear space $\cX$ is denoted by $M_{m, n}(\cX)$. For simplicity we denote $M_{m, n}(\C)$ by $M_{m, n}$. The (operator) norm of a matrix $A\in M_{m, n}$---when identified with an operator from $\ell_2^n$ into $\ell_2^m$---is denoted by  $\|A\|_{m,n}$.
	When $m=n$, we let $M_n(\cX) = M_{n, n}(\cX)$. 
	
	An element $T\in B(\cH)$ is positive (positive semidefinite) if it is of the form $T=X^*X$ for some $X\in B(\cH)$, where $X^*$ is the adjoint of $X$.  We shall write $T\geq0$ to denote that $T$ is positive. 
	
	For each $n\geq 1$, the space $M_n(B(\cH))$ of $n\times n$ matrices with entries in $B(\cH)$ is equipped with a norm naturally induced by $B(\cH)$. Indeed, there is a natural way to turn $M_n(B(\cH))$ into a $C^*$-algebra by identifying it with $B(\cH^n)$. To see this, let us denote a generic element of $M_n(B(\cH))$ by $[X_{ij}]$ where $X_{ij}\in B(\cH)$. Then, for $[X_{ij}],[Y_{ij}]$ in $M_n(B(\cH))$, we put
	$$
	[X_{ij}]\cdot [Y_{ij}]=\left[\sum_{k=1}^n X_{ik}Y_{kj}\right] \ \ \ \text{and} \ \ \ [X_{ij}]^*=[X_{ji}^*].
	$$
	Next, to introduce the norm on $M_n(B(\cH))$, we think of $X=[X_{ij}]\in M_n(B(\cH))$ as an element of $B(\mathcal{H}^n)$---which is algebraically isomorphic to $M_n(B(\mathcal{H}))$---and define 
	$$
	\|X\|_n:=\|X\|_{M_n(B(\cH))}=\sup\left\{\bigg(\,\sum_{i=1}^n\Big\|\sum_{j=1}^n X_{ij}(v_j)\Big\|^2\,\bigg)^{\frac1{2}}:v_1, \dots, v_n\in \mathcal{H},\sum_{j=1}^n \|v_j\|^2\leq 1\right\}.
	$$
	We can similarly define norms on $M_{m,n}(B(\cH))$.

	The above matrix norms satisfy two crucial properties. For $A, B^*\in M_{m, n}$ and $X\in M_n(B(\cH))$, by considering matrix multiplication in the usual way, we can think of $AXB$ as an element of $M_m(B(\cH))$.\footnote{Equivalently, $AXB$ can be understood as $(A\otimes I_{\cH})X (B\otimes I_{\cH})$ where $I_{\cH}$ is the identity operator acting on $\cH$.} Also, for $X\in M_m(B(\cH))$ and $Y\in M_n(B(\cH))$, we denote by $X\oplus Y$ the element of $M_{m+n}(B(\cH))$ given by 
	$$
	X\oplus Y = \begin{bmatrix}
		X & 0\\
		0 & Y
	\end{bmatrix}.
	$$
	Then, it is readily verified that the norms on $M_n(\cX)$, where $\cX=B(\cH)$, satisfy the following properties:
	\begin{itemize}
		\item[(i)] $\|AXB\|_m\leq \|A\|_{m,n}\|X\|_n\|B\|_{n,m}$ for all $A,B^*\in M_{m,n}$, $X\in M_n(\mathcal{X})$;
		
		\item[(ii)] $ \left\|X\oplus Y\right\|_{m+n}=\max\{\|X\|_n,\|Y\|_m\} $ for all $X\in M_n(\mathcal{X})$, $Y\in M_m(\mathcal{X})$.
		
	\end{itemize}

	We note that the above two properties hold not only in $B(\cH)$, but also in all of its linear subspaces. Conversely, Ruan's theorem~\cite{Pisier3} states that any Banach space $\mathcal{X}$, with a hierarchy of norms $\|\cdot\|_n$ on $M_n(\mathcal{X})$, $n\geq 1$, satisfying properties (i) and (ii)\footnote{In view of Ruan's theorem, these spaces are sometimes termed as \emph{abstract operator spaces} as opposed to concrete operator spaces (see Definition \ref{cos}).} ``embeds" into $B(\cH)$  for some Hilbert space $\cH$.
	
	\begin{definition}\label{cos}
		A (concrete) operator space is a subspace of $B(\cH)$ for some Hilbert space $\cH$ that is closed in operator norm topology.
	\end{definition}
	
	Ruan's theorem may be precisely formulated by using the notion of completely bounded maps. Let $\phi:\mathcal{X}\rightarrow \mathcal{Y}$ be a linear map between operator spaces $\mathcal{X},\mathcal{Y}$. The $n$-th \emph{amplification} $\phi_n:M_n(\mathcal{X})\rightarrow M_n(\mathcal{Y})$ of $\phi$ is defined by $\phi_n([X_{ij}])=[\phi(X_{ij})]$ for $[X_{ij}]\in M_n(\mathcal{X})$. Then $\phi$ is said to be \emph{completely bounded} if $\|\phi\|_{cb}:=\sup_{n\geq 1} \|\phi_n\|<\infty$. 
	We say that $\phi$ is a \emph{complete isometry} if each $\phi_n$ is an isometry.

	Using this terminology, by Ruan's theorem any Banach space $\mathcal{X}$ with a sequence of norms $\|\cdot\|_n$ on $M_n(\mathcal{X})$ for $n\geq 1$ satisfying properties (i) and (ii) above, is completely isometric to an operator space, i.e., by Definition~\ref{cos} it can be identified with a closed subspace of some $B(\cH)$ .

	\subsection{Dual operator space structure}
	There is a well-behaved notion of dual space in the category of operator spaces and completely bounded maps. Let $\mathcal{X},\mathcal{Y}$ be two operator spaces, and let $(CB(\cX, \cY),\|\cdot\|_{cb})$ be the Banach space of all completely bounded maps from $\cX$ into $\cY$. Then, $CB(\cX, \cY)$ may be endowed with a natural operator space structure which we now describe. To this end, we first identify $M_n(CB(\mathcal{X},\mathcal{Y}))$ with $CB(\mathcal{X},M_n(\mathcal{Y}))$ as follows: for $\Phi=[\phi_{ij}]\in M_n(CB(\mathcal{X},\mathcal{Y}))$, let $\widetilde{\Phi}:\mathcal{X}\rightarrow M_n(\mathcal{Y})$ be given by $\widetilde{\Phi}(X):=\big[\phi_{ij}(X)\big]$. Then we can define a norm on $M_n(CB(\mathcal{X},\mathcal{Y}))$ by letting 
	$$\|\Phi\|_n:=\big\|\widetilde{\Phi}\big\|_{CB(\mathcal{X},M_n(\mathcal{Y}))}.$$ 
	It is not hard to verify that $CB(\mathcal{X},\mathcal{Y})$ with this sequence of norms satisfies Ruan's axioms (properties (i) and (ii) above), and by Ruan's theorem is an operator space.
	
	Considering the special case $\cY=\C$, we find that $CB(\cX, \C)$ is an operator space. On the other hand, one may observe that $CB(\cX, \C)$ coincides with $\cX^*=B(\cX, \C)$ with identical norms as Banach spaces. Therefore, the Banach dual $\cX^*$ of any operator space has a canonical operator space structure.

	\subsection{Column and row operator spaces}\label{subsec:C-R-spaces}
	Here we present two important examples of operator spaces which will play pivotal roles in the construction of vector-valued Schatten spaces in the next section.
	
	Let $\cH_1, \cH_2$ be two Hilbert spaces.  Then the map
	\begin{align*}\label{eq:embedding-H1-H2}
		X\mapsto \begin{bmatrix}
			0 & 0\\
			X & 0
		\end{bmatrix}
	\end{align*}
	is an isometric embedding of $B(\cH_1, \cH_2)$ into $B(\cH_1\oplus \cH_2)$. Therefore, $B(\cH_1, \cH_2)$ is naturally endowed with an operator space structure. Note that the induced norm on $M_n\big(B(\cH_1, \cH_2)\big)$ coincides with the norm of $B(\cH_1^n, \cH_2^n)$. This means that to understand the operator space structure of $B(\cH_1, \cH_2)$ we can indeed ignore the above embedding and directly work with $B(\cH_1^n, \cH_2^n)$~\cite[p.~22]{Pisier3}.

	Considering spacial cases of the above construction, we find that both $B(\C, \cH)$ and $B(\cH, \C)$ are operator spaces. Indeed, for any Hilbert space $\cH$, the Banach spaces
	$$\cH_c := B(\C, \cH)\quad  \text{ and  } \quad \cH_r:=B(\cH^*, \C)$$
	have operator space structures. $\cH_c$ is called the \emph{column operator space} while $\cH_r$ is called the \emph{row operator space}~\cite[p.~25]{Pisier2}. 
	
	We note that elements of $\cH_c$ and $\cH_r$ can be identified with those of $\cH$. To this end, for any $v\in \cH$ define $v^c\in \cH_c$ and $v^r\in \cH_r$ by
	$$v^c(\alpha) = \alpha v, \qquad v^r(w)= \langle w, v\rangle, \qquad \alpha\in \C, w\in \cH^*.$$
	Then, $v^c\mapsto v^r$ provides an isometry between $\cH^c$ and $\cH^r$ as Banach spaces. Nevertheless, the operator space structures of $\cH^c$ and $\cH^r$ do not coincide under this identification. To verify this, suppose that  $n\leq \dim \cH$, and define $v_{ij}\in \cH$ for $1\leq i,j\leq n$ as follows: let $v_{ij}=0$ if $j\neq 1$ and let $\{v_{11}, \dots, v_{n1}\}\subset \cH$ be orthonormal.  Then we have
	\begin{align*}
		\big\|\big( v_{ij}^c \big)\big\|_{M_n(\cH_c)} &= \sup_{\stackrel{\alpha_1, \dots, \alpha_n\in \C}{\sum |\alpha_k|^2=1}} \Big\|   \Big(  \sum_{j} \alpha_jv_{ij}  \Big)   \Big\|_{\cH^n}\\
		&= \sup_{\stackrel{\alpha_1, \dots, \alpha_n\in \C}{\sum |\alpha_k|^2=1}} \big\|   \big(   \alpha_1v_{i1}  \big)   \big\|_{\cH^n}\\
		&= \big\|   \big(   v_{i1}  \big)   \big\|_{\cH^n}\\
		& = \Big(\sum_i \|v_{i1}\|^2_{\cH}\Big)^{1/2}\\
		& = \sqrt n,
	\end{align*}
	while
	\begin{align*}
		\big\|\big( v_{ij}^r \big)\big\|_{M_n(\cH_r)} &= \sup_{\stackrel{w_1, \dots, w_n\in \cH^*}{\sum \|w_k\|^2=1}}  \Big\|   \Big(  \sum_{j} \langle w_j, v_{ij}\rangle  \Big)   \Big\|\\
		&= \sup_{\stackrel{w_1, \dots, w_n\in \cH^*}{\sum \|w_k\|^2=1}}  \big\|   \big(   \langle w_1, v_{i1}\rangle  \big)   \big\|\\
		&= \sup_{\stackrel{w_1, \dots, w_n\in \cH^*}{\sum \|w_k\|^2=1}}  \Big(\sum_i   |\langle w_1, v_{i1}\rangle|^2  \Big)^{1/2}   \\
		&= \sup_{\stackrel{w_1, \dots, w_n\in \cH^*}{\sum \|w_k\|^2=1}} \|w_1\|_{\cH^*}\\
		&=1,
	\end{align*}
	where in the penultimate equality we use the fact that $\{v_{11}, \dots, v_{n1}\}$ is orthonormal.

	For simplicity of notation, we denote
	$$\cC=(\ell_2)_c, \quad \cR=(\ell_2)_r.$$
	In the finite-dimensional case of $\cH=\ell_2^n$, we denote the associated spaces by $\cC^n, \cR^n$.
	Note that $\cC$ and $\cR$ can be identified with the spaces of column and row vectors, respectively. In what follows, we denote elements of $\cC, \cR$ with superscripts $v^c, w^r$ to highlight that they are column and row vectors, respectively.
	
	\begin{remark}\label{rem:ell-2-*}
		In the following, we will also need $(\ell^*_2)_c, (\ell_2^*)_r$. Nevertheless, since these spaces can again be identified by column and row vectors, we abuse the above notations and denote them by $\cC, \cR$, respectively. 
		
	\end{remark}

	\subsection{Interpolation of operator spaces}
	
	Let $(\mathcal{X}_0,\mathcal{X}_1)$ be a compatible couple of Banach spaces. If further $\mathcal{X}_0,\mathcal{X}_1$ are operator spaces, then $M_n(\mathcal{X}_0)$ and $M_n(\mathcal{X}_1)$ with their specific norms form a compatible couple of Banach spaces. Thus, we may consider the complex interpolation space $[M_n(\mathcal{X}_0),M_n(\mathcal{X}_1)]_\theta$, for every $n\geq 1$, and endow $\mathcal{X}_\theta=[\mathcal{X}_0,\mathcal{X}_1]_\theta$ with a matrix norm structure as follows:
	\begin{equation*}
		\big\|[x_{ij}]\big\|_{M_n(\mathcal{X}_\theta)}:=	\big\|[x_{ij}]\big\|_{[M_n(\mathcal{X}_0),M_n(\mathcal{X}_1)]_\theta}, \ \ \ [x_{ij}]\in M_n(\mathcal{X}_\theta).
	\end{equation*}
	It can be shown that these norms satisfy Ruan's axioms, so $\mathcal{X}_\theta$ turns into an operator space~\cite{Pisier, Pisier3}. Moreover, we have the following analogue of Theorem~\ref{thm:interpolation-bound}.

	\begin{prop}[\cite{Pisier}]\label{operatorspaceinterpolation}
		Let $(\cX_0, \cX_1)$ and $(\cY_0, \cY_1)$ be a pair of compatible operator spaces, and define the operator spaces $\cX_\theta = [\cX_0, \cX_1]$ and $\cY_\theta = [\cY_0, \cY_1]$ as above.  
		Let  $T:\mathcal{X}_0+\mathcal{X}_1\rightarrow \mathcal{Y}_0+\mathcal{Y}_1$ be a linear map such that $T|_{\mathcal{X}_j}\in CB(\mathcal{X}_j,\mathcal{Y}_j)$, $j=0,1$. Then $T|_{\mathcal{X}_\theta}\in CB(\mathcal{X}_\theta,\mathcal{Y}_\theta)$, and
		\begin{equation*}
			\|T\|_{CB(\mathcal{X}_\theta,\mathcal{Y}_\theta)}\leq\|T\|_{CB(\mathcal{X}_0,\mathcal{Y}_0)}^{1-\theta}\|T\|_{CB(\mathcal{X}_1,\mathcal{Y}_1)}^{\theta}.
		\end{equation*}
	    Furthermore, the reiteration identity \eqref{eq:reiteration} holds completely isometrically for a compatible couple of operator spaces $\cX_0, \cX_1$.
	\end{prop}

	\begin{example}
		Let $\mathcal{H}$ be a Hilbert space, and let $\mathcal{S}_\infty(\cH)\subseteq B(\cH)$ be the space of compact operators on $\cH$ equipped with the operator norm. Note that $\mathcal{S}_\infty(\cH)$, as a closed subspace of $B(\cH)$, is an operator space.

		Let $\mathcal{S}_1(\cH)$ be the space of trace class operators acting on $\cH$ equipped with the norm
		$$\|X\|_{\mathcal{S}_1(\cH)} = \tr|X|,$$
		where $|X| = \sqrt{X^*X}$.
		It is well-known that $\mathcal{S}_1(\cH)$ is the (Banach) dual of $\mathcal{S}_\infty(\cH)$.\footnote{Note that the dual of $\mathcal{S}_1(\cH)$ is $B(\cH)$, which does not coincide with $\mathcal{S}_\infty(\cH)$ when $\cH$ is not finite-dimensional.} Therefore, since the dual of any operator space is an operator space, $\mathcal{S}_1(\cH)$ inherits an operator space structure.  
		
		The Schatten class $\mathcal{S}_p(\cH)$ is indeed defined for any $p\geq 1$; let $\mathcal{S}_p(\cH)$ be the space of operators $X$ in $B(\cH)$ with $\tr |X|^p<\infty$, equipped with the norm 
		$$\|X\|_{\mathcal{S}_p(\cH)} = \big(\tr|X|^p\big)^{\frac{1}{p}}.$$
		When $\cH=\ell_2$ (resp. $\cH=\ell_2^n$) we denote $\mathcal{S}_p(\cH)$ by $\mathcal{S}_p$ (resp. $\mathcal{S}_p^n$). The space $\mathcal{S}_p$ is indeed the noncommutative counterpart of the ordinary $L_p$-space (in the discrete case). Moreover,~\eqref{eq:interpolation-Lp} is generalized to Schatten classes with similar ideas~\cite{Pisier3}. In particular, we have
		\begin{align}\label{eq:Sp-interpolation}
			\mathcal{S}_p(\cH) =\big[\mathcal{S}_\infty(\cH), \mathcal{S}_1(\cH)\big]_{\frac{1}{p}}, \qquad 1< p<\infty,
		\end{align}
		with identical norms.
	\end{example}

	\begin{remark}\label{remark3.5}
		We argued that $\mathcal{S}_1(\cH), \mathcal{S}_\infty(\cH)$ are operator spaces, so if \eqref{eq:Sp-interpolation} is interpreted as an identity between operator spaces, then $\mathcal{S}_p(\cH)$ for any $1\leq p\leq \infty$ is endowed with an operator space structure. Moreover, by using reiteration and \eqref{eq:Sp-interpolation} we obtain that $\mathcal{S}_{p_\theta}=[\mathcal{S}_{p_0},\mathcal{S}_{p_1}]_\theta$ completely isometrically, which is the noncommutative analogue of \eqref{eq:interpolation-Lp} in the discrete case.
	\end{remark}

	\subsection{Haagerup tensor product and interpolation}
	Haagerup tensor product provides a special operator space structure for the tensor product of two operator spaces. Its construction is an important step towards understanding operator-valued Schatten spaces and is reviewed in this subsection.
	
	Let $X=[X_{ij}]\in M_{n,m}(\mathcal{X})$ and $Y=[Y_{ij}]\in M_{m, n}(\mathcal{Y})$, where $\mathcal{X},\mathcal{Y}$ are two operator spaces. Then the element $X\odot Y\in M_{n}(\mathcal{X}\otimes \mathcal{Y})$ is defined by
	$$
	X\odot Y:=\left[\sum_{k=1}^{m} X_{ik}\otimes Y_{kj}\right].
	$$
	Using this notation for $Z\in M_n(\mathcal{X}\otimes \mathcal{Y})$, we define its \emph{Haagerup matrix norm} by
	$$
	\|Z\|_{h,n}:=\inf\big\{\|X\|_{n,m}\|Y\|_{m,n}:Z=X\odot Y\big\},
	$$
	where the infimum is taken over all $m$ and all $X,Y$ as above.

	We verify that $\|\cdot\|_{h,n}$ satisfies the triangle inequality and is thus an actual norm~\cite{EffrosKishimoto}. For $Z_i\in M_n(\cX\otimes \cY)$, $i=1, 2$, and arbitrary $\epsilon>0$, let $X_i\in M_{n, m_{i}}(\cX)$ and $Y_i\in M_{m_i, n}(\cY)$ be such that $Z_i= X_i\odot Y_i$ and 
	\begin{align}\label{eq:Xi-Yi-Zi-esp}
		\|X_i\|_{n, m_i}\|Y_i\|_{m_i, n}\leq \|Z_i\|_{h, n} + \epsilon.
	\end{align}
	After scaling $X_i, Y_i$, we can assume that $\|X_i\|_{n, m_i}=\|Y_i\|_{m_i, n}$. Next, define
	$$\widetilde X = \big[X_1, X_2\big]\in M_{n, m_1+m_2}(\cX), \qquad \widetilde Y = \begin{bmatrix}
		Y_1\\
		Y_2
	\end{bmatrix}\in M_{m_1+m_2, n}(\cY).$$
	Then, we have $Z_1+Z_2= \widetilde X\odot \widetilde Y$, and
	\begin{align*}
		\|Z_1+Z_2\|_{h, n} &\leq  \big\|\widetilde X\big\|_{n, m_1+m_2} \big\|\widetilde Y\big\|_{m_1+m_2, n}\\
		&=  \big\|\widetilde X \widetilde X^* \big\|_{n}^{1/2} \big\|\widetilde Y^* \widetilde Y\big\|_{n}^{1/2}\\
		&=  \big\| X_1X_1^* +X_2X_2^* \big\|_{n}^{1/2} \big\|Y^*_1Y_1 +Y^*_2Y_2\big\|_{n}^{1/2}\\
		&\leq \frac 12\Big( \big\| X_1X_1^* +X_2X_2^* \big\|_{n} + \big\|Y^*_1Y_1 +Y^*_2Y_2\big\|_{n}\Big)\\
		&\leq \frac 12\Big( \big\| X_1X_1^*\big\|_n +\big\|X_2X_2^* \big\|_{n} + \big\|Y^*_1Y_1\big\|_n +\big\|Y^*_2Y_2\big\|_{n}\Big)\\
		&\leq \frac 12\Big( \big\| X_1\big\|^2_n +\big\|X_2\big\|^2_{n} + \big\|Y_1\big\|^2_n +\big\|Y_2\big\|^2_{n}\Big)\\
		&\leq \|Z_1\|_{h, n} + \|Z_2\|_{h, n}+2\epsilon,
	\end{align*}
	where in the last line we use~\eqref{eq:Xi-Yi-Zi-esp} and  $\|X_i\|_{n, m_i}=\|Y_i\|_{m_i, n}$. Then, $\|Z_1+Z_2\|_{h, n}\leq \|Z_1\|_{h, n} + \|Z_2\|_{h, n}$ since $\epsilon>0$ is arbitrary.

	It can also be verified that this matrix norm structure satisfies Ruan's axioms. Thus, after completion in this norm, we obtain an operator space structure on $\cX\otimes \cY$ called the  \emph{Haagerup tensor product} which we denote by $\mathcal{X}\otimes_{h} \mathcal{Y}$.

	The Haagerup tensor product is associative, that is, $(\cX\otimes_h\cY)\otimes_h \mathcal Z=\cX\otimes_h(\cY\otimes_h \mathcal Z)$ completely isometrically. Nonetheless, it is \emph{not} commutative as in general $\cX\otimes_h\cY$ and $\cY\otimes_h\cX$ carry different norms.
	
	\medskip
	In the following proposition we collect some properties of Haagerup tensor product in regard to column and row Hilbert spaces. This proposition is a key step towards understanding the construction of operator-valued Schatten spaces in the next section. Its proof is included in Appendix \ref{proofs}.
	
	\begin{prop}\label{prop:h-tensor-R-C}
		Let $\cH, \cK$ be two Hilbert spaces. Then,
		\begin{itemize}
			\item[{\rm (i)}] we have the complete isometries
			\begin{equation*}\label{eq:cc-rr-h-tensor}
				\mathcal{H}_c\otimes_h \mathcal{K}_c =   (\mathcal{H}\otimes_2\mathcal{K})_c   \quad \text{ and } \quad \mathcal{K}_r\otimes_h\mathcal{H}_r
				=(\mathcal{H}\otimes_2\mathcal{K})_r,
			\end{equation*}

			\item[{\rm (ii)}] and
			\begin{equation*}\label{eq:cr-rc-h-tensor}
				\mathcal{S}_\infty(\cH) = \cH_c\otimes_h (\cH^*)_r  \quad \text{ and } \quad \mathcal{S}_1(\cH) = \cH_r\otimes_h (\cH^*)_c.
			\end{equation*}	
			In particular, $\mathcal{S}_\infty = C\otimes_h \cR$ and $\mathcal{S}_1=\cR\otimes_h \cC$ completely isometrically.
			
			\item[{\rm (iii)}] Moreover, $M_{1, d}(\cC^d)= M_{d, 1}(\cR^d) = \mathcal{S}_\infty^d$ and $M_{1, d}(\cR^d) = M_{d, 1}(\cC^d) = \mathcal{S}_2^d$ isometrically.
			
		\end{itemize}
	\end{prop}

	Let $(\mathcal{X}_0,\mathcal{X}_1)$ and $(\mathcal{Y}_0,\mathcal{Y}_1)$ be two compatible couples of operator spaces. Then it can be verified that
	$(\mathcal{X}_0\otimes_h\mathcal{Y}_0,\mathcal{X}_1\otimes_h\mathcal{Y}_1)$ is also a compatible couple of operator spaces. The following theorem states that an interpolation space obtained from this couple can be written as the Haagerup tensor product of the interpolation spaces obtained from respective components. This theorem plays a pivotal role in identifying the norms of operator-valued Schatten spaces as variational expressions.

	\begin{theorem}[\cite{Pisier, Kouba}]\label{KouPis}
		Let $(\mathcal{X}_0,\mathcal{X}_1)$ and $(\mathcal{Y}_0,\mathcal{Y}_1)$ be two compatible couples of operator spaces. Then, for every $0<\theta<1$, we have a complete isometry
		\begin{equation}\label{eq:KouPis}
			[\mathcal{X}_0\otimes_h\mathcal{Y}_0,\mathcal{X}_1\otimes_h\mathcal{Y}_1]_\theta=[\mathcal{X}_0,\mathcal{X}_1]_\theta \otimes_h [\mathcal{Y}_0,\mathcal{Y}_1]_\theta.
		\end{equation}
		$($i.e., the Haagerup tensor product and interpolation are commuting operations.$)$
	\end{theorem}
	
	The proof that follows is essentially borrowed from~\cite{Kouba} and relies heavily on a variant of the Wiener--Masani theorem which is a factorization method for operator-valued functions tailored to our purposes (Appendix~\ref{sec:Wiener--Masani}). The other major ingredients of the proof are a \emph{subharmonicity lemma} and a \emph{selection lemma} which we state and prove in Appendix~\ref{sec:Subharmonicity}. These are mostly technical details, so the reader is advised to first read the following proof to acquaint themself with the main ideas, and then, if interested, proceed to Appendices~\ref{sec:Subharmonicity} and~\ref{sec:Wiener--Masani}. The latter appendix on the Wiener--Masani theorem contains ideas that are independent of the rest of the paper and is based on~\cite{Helson}.
	
	\begin{proof}
		To simplify notation, we denote $\mathcal{X}_\theta=[\mathcal{X}_0,\mathcal{X}_1]_\theta$, $\mathcal{Y}_\theta=[\mathcal{Y}_0,\mathcal{Y}_1]_\theta$ and $\mathcal{Z}_\theta=[\mathcal{Z}_0,\mathcal{Z}_1]_\theta$ where $\mathcal{Z}_j=\mathcal{X}_j\otimes_{h}\mathcal{Y}_j$, $j=0,1$. Then \eqref{eq:KouPis} is tantamount to
		$$
		\|W\|_{M_n(\mathcal{Z}_\theta)}=\|W\|_{M_n(\mathcal{X}_\theta\otimes_{h}\mathcal{Y}_\theta)}, \ \ \ n\geq1.
		$$
		We assume that all the spaces are finite-dimensional,\footnote{For a generalization to infinite dimensions the reader is referred to \cite{Pisier}.} and prove this identity only for $n=1$, as the generalization to $n>1$ is straightforward.
		
		Let us first show that $\|W\|_{\mathcal{Z}_\theta}\leq\|W\|_{\mathcal{X}_\theta\otimes_{h}\mathcal{Y}_\theta}$. Let $\epsilon>0$ be arbitrary. By definition we have
		\begin{equation*}
			\|W\|_{\mathcal{X}_\theta\otimes_{h}\mathcal{Y}_\theta}=\inf\big\{\|X\|_{M_{1,m}(\mathcal{X}_\theta)}\cdot\|Y\|_{M_{m,1}(\mathcal{Y}_\theta)}:W=X\odot Y\big\},
		\end{equation*}
	where $X\in M_{1,m}(\mathcal{X}_\theta)$, $Y\in M_{m,1}(\mathcal{Y}_\theta)$, and
		\begin{equation*}
			\|X\|_{M_{1,m}(\mathcal{X}_\theta)}=	\|X\|_{[M_{1,m}(\mathcal{X}_0),M_{1,m}(\mathcal{X}_1)]_\theta}.
		\end{equation*}
		Therefore, by the definition of the norm in the interpolation space, there is a holomorphic function $f(z)$ such that $f(\theta)=X$ and
		\begin{equation}\label{eqa}
			\Big( \sup_t \|f(it)\|_{M_{1,m}(\mathcal{X}_0)}\Big)^{1-\theta} \Big(\sup_t \|f(1+it)\|_{M_{1,m}(\mathcal{X}_1)}   \Big)^\theta<(1+\epsilon)\|X\|_{M_{1,m}(\mathcal{X}_\theta)}.
		\end{equation}
		Similarly, there is a holomorphic function $g(z)$ such that $g(\theta)=Y$ and
		\begin{equation}\label{eqb}
			\Big( \sup_t \|g(it)\|_{M_{m,1}(\mathcal{Y}_0)}\Big)^{1-\theta} \Big(\sup_t \|g(1+it)\|_{M_{m,1}(\mathcal{Y}_1)}   \Big)^\theta<(1+\epsilon)\|Y\|_{M_{m,1}(\mathcal{Y}_\theta)}.
		\end{equation}
		Then, the function $h(z):=f(z)\odot g(z)$ is holomorphic, $h(\theta)=W$, and
		\begin{equation}\label{eqc}
			\|W\|_{\mathcal{Z}_\theta}\leq \Big( \sup_t \|h(it)\|_{\mathcal{Z}_0}\Big)^{1-\theta} \Big(\sup_t \|h(1+it)\|_{\mathcal{Z}_1}   \Big)^\theta.
		\end{equation}
		Since $\mathcal{Z}_j=\mathcal{X}_j\otimes_{h}\mathcal{Y}_j$, $j=0,1$, by the definition of the Haagerup norm we have
		\begin{equation}\label{eqd}
			\|h(j+it)\|_{\mathcal{Z}_j}\leq \|f(j+it)\|_{M_{1,m}(\mathcal{X}_j)}\cdot \|g(j+it)\|_{M_{m,1}(\mathcal{Y}_j)}.
		\end{equation}
		Finally, combining \eqref{eqa}, \eqref{eqb}, \eqref{eqc}, \eqref{eqd} and letting $\epsilon\rightarrow0$ yields 
		$$\|W\|_{\mathcal{Z}_\theta}\leq \|X\|_{M_{1,m}(\mathcal{X}_\theta)}\cdot \|Y\|_{M_{m,1}(\mathcal{Y}_\theta)},$$ 
		which in turn gives us the desired inequality upon minimizing over all $X,Y$ as above.
		
		\medskip
		To prove the reversed inequality, we assume without loss of generality that $\mathcal{X}_0,\mathcal{X}_1$ have the same dimension, say $d$, and identify them with $M_{d,1}$ as vector spaces. Likewise, we identify $\mathcal{Y}_0,\mathcal{Y}_1$ with $M_{1,d'}$ as vector spaces with the same dimension $d'$. Then, considering the canonical identifications
		\begin{equation*}
			M_{1,m}(\mathcal{X}_b)\cong M_{1,m}\otimes M_{d,1}\cong M_{d,m},
		\end{equation*}
		and
		\begin{equation*}
			M_{m,1}(\mathcal{Y}_b)\cong M_{m,1}\otimes M_{1,d'}\cong M_{m,d'},
		\end{equation*}
		we observe that for $\widetilde{X}\in M_{1,m}(\mathcal{X}_b)$ and $\widetilde{Y}\in M_{m,1}(\mathcal{Y}_b)$, we have $\widetilde{X}\odot \widetilde{Y}=X Y$ where $X\in M_{d,m}$ and $Y\in M_{m,d'}$ correspond, respectively, to $\widetilde{X}$ and $\widetilde{Y}$ via the above-mentioned identifications. In the rest of the proof, we assume  that $d\leq d'$. Also, applying the polar decomposition we assume without loss of generality that $m=d$.
		
		Let $\epsilon>0$ be arbitrary. Then, by definition there is a holomorphic function $h(z)$ such that $h(\theta)=W$ and
		\begin{equation}\label{ff}
			\max\big\{\sup_t \|h(b+it)\|_{\mathcal{Z}_b}:b=0,1\big\}<(1+\epsilon)\|W\|_{\mathcal{Z}_\theta}.
		\end{equation}
		Note that
		\begin{equation*}
			\|h(b+it)\|_{\mathcal{Z}_b}=\inf\big\{\|A\|_{M_{1,d}(\mathcal{X}_b)}\cdot \|B\|_{M_{d,1}(\mathcal{Y}_b)}:\, h(b+it)=A B,A\in M_{d,d},B\in M_{d,d'}\big\},
		\end{equation*}
		where as mentioned above we have identified $M_{1,d}(\mathcal{X}_b)\cong M_{d, d}$ and $M_{d,1}(\mathcal{Y}_b)\cong M_{d,d'}$.
		Then, by applying an appropriate \emph{selection theorem} as in Lemma~\ref{selection lemma}, we obtain bounded measurable functions $A^{(b)}(b+it),B^{(b)}(b+it)$ such that for every $t\in \mathbb{R}$ and $b=0,1$,
		\begin{equation}\label{K1}
			h(b+it)=A^{(b)}(b+it) B^{(b)}(b+it),
		\end{equation}
		\begin{equation}\label{K2}
			\|A^{(b)}(b+it)\|_{M_{1,d}(\mathcal{X}_b)}\leq 1+\frac{\epsilon}{2},
		\end{equation}
		\begin{equation}\label{K3}
			\|B^{(b)}(b+it)\|_{M_{d,1}(\mathcal{Y}_b)}\leq \|h(b+it)\|_{\mathcal{Z}_b}.
		\end{equation}
		Next, using the Wiener--Masani theorem (see Corollary~\ref{corr:Wiener--Masani}) we can find a bounded holomorphic function $F(z)$ such that for almost every $t\in \mathbb{R}$,
		\begin{equation}\label{gg}
			F(b+it) F(b+it)^*=A^{(b)}(b+it) A^{(b)}(b+it)^*+\frac{\epsilon^2}{4\|I\|_{M_{1,d}(\mathcal{X}_b)}^2}I.
		\end{equation}
		It follows from \eqref{K2}, \eqref{gg} and the $C^*$-property of the norms that for almost every $t\in \mathbb{R}$,
		\begin{equation}\label{kk}
			\|F(b+it)\|_{M_{1,d}(\mathcal{X}_b)}\leq \left(\|A^{(b)}(b+it)\|_{M_{1,d}(\mathcal{X}_b)}^2+\frac{\epsilon^2}{4}\right)^{\frac{1}{2}}\leq 1+\epsilon.
		\end{equation}
		We note that by the Wiener--Masani theorem $F(z)$ is invertible almost everywhere. Thus, we can define measurable functions $S^{(b)}(\cdot)$, $b=0,1$, such that
		\begin{equation*}
			S^{(b)}(b+it)=F(b+it)^{-1} A^{(b)}(b+it),
		\end{equation*}
		for almost every $t$.
		Then, $S^{(b)}(b+it)S^{(b)}(b+it)^*\leq I$ for almost every $t$ since $A^{(b)}(b+it) A^{(b)}(b+it)^*\leq F(b+it) F(b+it)^*$.
		
		Now, define the holomorphic function $G(z)$ on $\mathbb S$ by
		$$G(z):=F(z)^{-1}\cdot h(z).$$  
		By~\eqref{K1} and the preceding paragraph we obtain
		\begin{align*}
			G(b+it) &=F(b+it)^{-1}\cdot A^{(b)}(b+it) B^{(b)}(b+it)\\
			&=S^{(b)}(b+it) B^{(b)}(b+it).
		\end{align*}
		This implies that $G(z)$ is bounded, and by~\eqref{K3} for almost every $t$ we have
		\begin{equation}\label{kkk}
			\|G(b+it)\|_{M_{d,1}(\mathcal{Y}_b)} \leq \|B^{(b)}(b+it)\|_{M_{d,1}(\mathcal{Y}_b)}\leq \|h(b+it)\|_{\mathcal{Z}_b}.
		\end{equation}
		Finally, using \eqref{ff}, \eqref{kk}, \eqref{kkk}, and part (ii) of Lemma \ref{lem:sub-harmonicity}, together with the identities just above the lemma, we obtain
		\begin{align*}
			\|W\|_{\mathcal{X}_\theta\otimes_{h}\mathcal{Y}_\theta} &=\|h(\theta)\|_{\mathcal{X}_\theta\otimes_{h}\mathcal{Y}_\theta}\\
			&\leq \|F(\theta)\|_{M_{1,d}(\mathcal{X}_b)} \|G(\theta)\|_{M_{d,1}(\mathcal{Y}_b)}\\
			&\leq \exp \sum_{b}\int P_b(\theta,s)\log \Big(\|F(b+is)\|_{M_{1,d}(\mathcal{X}_b)} \|G(b+is)\|_{M_{d,1}(\mathcal{Y}_b)} \Big)\dd s\\
			&\leq \exp \sum_{b}\int P_b(\theta,s)\log\Big( (1+\epsilon) \|h(b+is)\|_{\mathcal{Z}_b}\Big) \dd s\\
			&\leq \exp \sum_{b}\int P_b(\theta,s)\log\bigg( (1+\epsilon)\Big(\sup_t \|h(b+it)\|_{\mathcal{Z}_b}\Big)\bigg) \dd s\\
			&=(1+\epsilon)\Big( \sup_t \|h(it)\|_{\mathcal{Z}_0}\Big)^{1-\theta} \Big(\sup_t \|h(1+it)\|_{\mathcal{Z}_1}   \Big)^\theta\\
			&\leq (1+\epsilon)\max\big\{\sup_t\|h(b+it)\|_{\mathcal{Z}_b}:b=0,1\big\}\\
			&\leq (1+\epsilon)^2\|W\|_{\mathcal{Z}_\theta}.
		\end{align*}
		Letting $\epsilon\rightarrow0$, the desired inequality is obtained.
	\end{proof}

	\section{Operator-valued Schatten spaces}\label{sec:NCVV}
	
	In this section, building upon the previous constructions and results, we finally introduce the family $\mathcal{S}_p[\cX]$, $1\leq p\leq\infty$, for an operator space $\cX$. Indeed, we first introduce the endpoint spaces $\mathcal{S}_\infty[\mathcal{X}]$ and $\mathcal{S}_1[\mathcal{X}]$ as certain Haagerup tensor products of $\mathcal{X}$ with the column and row operator spaces, and define $\mathcal{S}_p[\cX]$, $1< p<\infty$, to be an interpolation space between them. Then, by employing Theorem \ref{KouPis}, we establish an identification of $\mathcal{S}_p[\cX]$ which allows us to view $\mathcal{S}_\infty[\mathcal{X}]$, $\mathcal{S}_1[\mathcal{X}]$ and $\mathcal{S}_p[\cX]$ in a unified way. Next, we make use of this identification to show that the family $\mathcal{S}_p[\mathcal{S}_q]$, $1\leq p,q\leq\infty$, is an interpolation scale. In the final part of this section we establish variational characterizations of the norms of $\mathcal{S}_p[\cX]$ and $\mathcal{S}_p[\mathcal{S}_q]$.

	\subsection{Definition and basic properties}
	
	For an operator space $\mathcal{X}$, we first put
	\begin{equation}\label{SinftySone}
		\mathcal{S}_\infty[\mathcal{X}]:=\mathcal{C}\otimes_h\mathcal{X}\otimes_h \mathcal{R} \ \ \ \ \text{and} \ \ \ \ \mathcal{S}_1[\mathcal{X}]:=\mathcal{R}\otimes_h\mathcal{X}\otimes_h \mathcal{C}.
	\end{equation}

	These definitions are consistent with, and indeed motivated by, the identities $\mathcal{S}_\infty = C\otimes_h \cR$ and $\mathcal{S}_1=\cR\otimes_h \cC$ (see Proposition \ref{prop:h-tensor-R-C}(ii)). In fact, by taking $\mathcal{X}=\C$ we get $\mathcal{S}_\infty[\C]=\mathcal{S}_\infty$ and $\mathcal{S}_1[\C]=\mathcal{S}_1$.
	
	Recalling the notation developed in Subsection~\ref{subsec:C-R-spaces}, if $\{e_i^c:\, i\} $ and $\{e^r_i:\, i\}$ are the orthonormal bases of $\cC^n$ and $\cR^n$, respectively, corresponding to the standard basis of $\ell_2^n$, then one can identify $\mathcal{S}_\infty^n[\mathcal{X}]$ with $M_n(\cX)$ via the natural completely isometric map
	\begin{equation}\label{iddd}
		[X_{ij}]\mapsto\sum_{i,j}e_i^c\otimes X_{ij}\otimes e_j^r.
	\end{equation}
    One can also identify $\mathcal{S}_1^n[\mathcal{X}]$ with $M_n(\cX)$ similarly. However, this latter identification will not be a complete isometry.
	
	\begin{definition}\label{vector valued Schatten}
		Let $1<p<\infty$. The {\it $\mathcal{X}$-valued $p$-Schatten space} $\mathcal{S}_p[\mathcal{X}]$ is defined to be the operator space
		\begin{equation*}
			\mathcal{S}_p[\mathcal{X}]:=\big[\mathcal{S}_\infty[\mathcal{X}],\mathcal{S}_1[\mathcal{X}]\big]_{\frac{1}{p}}.
		\end{equation*}
		$\mathcal{S}_p^n[\cX]$ for $n\geq 1$ is defined similarly.
	\end{definition}

In the following theorem, we obtain an identification of 	$\mathcal{S}_p[\mathcal{X}]$ for $1<p<\infty$ that complements the identities in \eqref{SinftySone}, and allows us to treat $\mathcal{S}_\infty[\mathcal{X}]$, $\mathcal{S}_1[\mathcal{X}]$ and $\mathcal{S}_p[\mathcal{X}]$ ($1<p<\infty$) on an equal footing. Before proceeding further, we introduce a piece of notation:
\begin{equation*}\label{eq:def-R-theta}
	\mathcal{R}(\theta):=[\mathcal{R},\mathcal{C}]_\theta=[\mathcal{C},\mathcal{R}]_{1-\theta}, \ \ \ \mathcal{R}(0):=\mathcal{R}, \ \ \ \mathcal{R}(1):=\mathcal{C}.
\end{equation*}

\begin{theorem}\label{identification}
	Let $1<p<\infty$. Then we have completely isometrically
	\begin{equation*}
		\mathcal{S}_p[\mathcal{X}]=\mathcal{R}\big(1-\tfrac1{p}\big)\otimes_h\mathcal{X}\otimes_h\mathcal{R}\big(\tfrac1{p}\big).
	\end{equation*}
\end{theorem}
\begin{proof}
	By associativity of the Haagerup tensor product and an iterated use of Theorem \ref{KouPis}, we obtain
	\begin{align*}
			\mathcal{S}_p[\mathcal{X}]
		&=[\mathcal{C}\otimes_h\mathcal{X}\otimes_h \mathcal{R},\mathcal{R}\otimes_h\mathcal{X}\otimes_h \mathcal{C}]_{\frac{1}{p}}\\
		&=[\mathcal{C},\mathcal{R}]_{\frac{1}{p}}\otimes_h[\mathcal{X}\otimes_h \mathcal{R},\mathcal{X}\otimes_h \mathcal{C}]_{\frac{1}{p}}\\
		&=[\mathcal{C},\mathcal{R}]_{\frac{1}{p}}\otimes_h\mathcal{X}\otimes_h[\mathcal{R},\mathcal{C}]_{\frac{1}{p}}\\
		&=\mathcal{R}\big(1-\tfrac1{p}\big)\otimes_h\mathcal{X}\otimes_h\mathcal{R}\big(\tfrac1{p}\big),
	\end{align*}
as desired.
\end{proof}

Note that for $\cX=\C$, this result yields the completely isometric identity
\begin{align*}
	\mathcal{S}_p=\mathcal{R}\big(1-\tfrac1{p}\big)\otimes_h\mathcal{R}\big(\tfrac1{p}\big),
\end{align*}
which is the analogue of Proposition~\ref{prop:h-tensor-R-C}(ii) for $1<p<\infty$. Note also that one can view the elements of $\mathcal{S}_p[\mathcal{X}]$ as $\mathcal{X}$-valued matrices via a similar identification as in \eqref{iddd}. However, the identification will not be a complete isometry unless $p=\infty$.

In the next result, we collect some basic properties of the operator-valued Schatten spaces which will be used in the sequel. For the sake of completeness, the proofs are given in Appendix~\ref{proofs}.

\begin{prop}\label{interpolation}
	Let $\mathcal{X},\mathcal{X}_0,\mathcal{X}_1$ be operator spaces, and $1\leq p_0,p_1,q_0,q_1\leq\infty$. Then the following statements hold true.
	
	(i) Interpolation: If $(\mathcal{X}_0,\mathcal{X}_1)$ is a compatible couple, then we have completely isometrically
	\begin{equation}\label{Pis1}
		\mathcal{S}_{p_\theta}[\mathcal{X}_\theta]=\big[\mathcal{S}_{p_0}[\mathcal{X}_0],\, \mathcal{S}_{p_1}[\mathcal{X}_1]\big]_\theta,
	\end{equation}
	where $p_\theta$ is given by $\frac1{p_{\theta}}=\frac{1-\theta}{p_0}+\frac{\theta}{p_1}$ and $\cX_\theta =[\cX_0, \cX_1]_\theta$.
	In particular, if $q_\theta$ is given by $\frac1{q_{\theta}}=\frac{1-\theta}{q_0}+\frac{\theta}{q_1}$, we have completely isometrically
	\begin{equation}\label{Pis2}
		\mathcal{S}_{p_\theta}[\mathcal{S}_{q_\theta}]=\big[\mathcal{S}_{p_0}[\mathcal{S}_{q_0}],\, \mathcal{S}_{p_1}[\mathcal{S}_{q_1}]\big]_\theta.
	\end{equation}

    (ii) Duality: Let $1<p\leq\infty$ and $p'$ be given by $1=\frac{1}{p}+\frac{1}{p'}$. Then we have completely isometrically
    \begin{equation*}
	\mathcal{S}_p[\mathcal{X}]^*=\mathcal{S}_{p'}[\mathcal{X}^*].
    \end{equation*}

    (iii) Fubini's theorem: For $1\leq p\leq\infty$, we have completely isometrically
    \begin{equation*}
    	\mathcal{S}_p[\mathcal{S}_p]=\mathcal{S}_p(\ell_2\otimes_2\ell_2).
    \end{equation*}
\end{prop}

	\subsection{Pisier's variational formula} 
	Here, we finally  prove the promised  characterization of the norm in the space $\mathcal{S}_p[\mathcal{S}_q]$. Its proof is based on the so-called Pisier's variational formula.

	\begin{theorem}[Pisier's formula \cite{Pisier2}]\label{Pisier's formula}
		Let $1\leq p<\infty$. Then for every $T\in \mathcal{S}_p[\mathcal{X}]$ we have
		\begin{equation*}
			\|T\|_{\mathcal{S}_p[\mathcal{X}]}=\inf\|A\|_{\mathcal{S}_{2p}}\|X\|_{\mathcal{S}_\infty[\mathcal{X}]}\|B\|_{\mathcal{S}_{2p}},
		\end{equation*}
		where the infimum is taken over all $A,B\in \mathcal{S}_{2p}$ and $X\in \mathcal{S}_\infty[\mathcal{X}]$ such that $T=A X B$.  As explained above, we understand the identity $T=A X B$ by thinking of $T, X$ as $\cX$-valued matrices.
	\end{theorem}
\begin{proof}
	Considering $\mathcal{S}_p^d[\cX]$ that is easier to work with,	we note that $\|X\|_{\mathcal{S}^d_\infty(\cX)}=\|X\|_{M_d(\cX)}$.
	Then, using Theorem~\ref{identification} and the definition of Haagerup tensor product, we need to evaluate $\|A\|_{M_{1,n}(\cR^d(1-\theta))}$ and $\|B\|_{M_{n, 1}(\cR^d(\theta))}$, where $\theta=\frac{1}{p}$. We first pad $A$ with zeros if necessary in order to assume that $n=d$. 
	Next, by the definition of the interpolation of operator spaces we have
	\begin{align*}
		M_{1,n}(\cR^n(1-\theta)) &= \big[ M_{1,n}(\cR^n), M_{1,n}(\cC^n)\big]_{1-\theta} \\
		& =   \big[ \mathcal{S}_2^n, \mathcal{S}_\infty^n\big]_{1-\theta} \\
		& = \mathcal{S}_{2p}^n,
	\end{align*}
	where the second equality follows from  part (iii) of Proposition~\ref{prop:h-tensor-R-C}, and the third equality follows from~\eqref{eq:Sp-interpolation} and the reiteration theorem. Therefore, $\|A\|_{M_{1,n}(\cR^n(1-\theta))} =\|A\|_{\mathcal{S}_{2p}^n}$.  Similarly, it can shown that $\|B\|_{M_{n, 1}(\cR^n(\theta))}$ equals the $2p$-norm of $B$. Then the desired result follows.
\end{proof}

	\begin{theorem}[\cite{Pisier2, Junge96}]\label{thm:Pisier-Junge-formula}
		Let $1\leq p\leq q\leq \infty$ and define $r$ by 
		$$\frac{1}{p}=\frac{1}{q}+\frac{1}{r}.$$
		Then we have the following:
		\begin{itemize}
			\item[{\rm (i)}] For every $X$ in $\mathcal{S}_p[\mathcal{S}_q]$ define
			\begin{equation}\label{eq:(p,q)-norm}
				\|X\|_{(p,q)}:=\inf\|A\|_{\mathcal{S}_{2r}}\|Y\|_{\mathcal{S}_q(\ell_2\otimes_2\ell_2)}\|B\|_{\mathcal{S}_{2r}},
			\end{equation}
			where the infimum is taken over all $A,B\in \mathcal{S}_{2r}$ and $Y\in \mathcal{S}_q(\ell_2\otimes_2\ell_2)$ such that $X=A Y B$. Then we have $\|X\|_{(p, q)}=\|X\|_{\mathcal{S}_p[\mathcal{S}_q]}$. We note that the multiplication $AYB$ is understood as before, and indeed is equal to $(A\otimes I_{\ell_2}) Y (B\otimes I_{\ell_2})$, where $I_{\ell_2}\in  B(\ell_2)$ is the identity operator.

			\item[{\rm (ii)}] For every $X$ in $\mathcal{S}_q[\mathcal{S}_p]$ define
			\begin{equation}\label{eq:(q,p)-norm}
				\|X\|_{(q,p)}:=\sup\|A X B\|_{\mathcal{S}_p(\ell_2\otimes_2\ell_2)},
			\end{equation}
			where the supremum is taken over all $A,B\in \mathcal{S}_{2r}$ such that $\|A\|_{\mathcal{S}_{2r}},\|B\|_{\mathcal{S}_{2r}}\leq 1$. Then we have $\|X\|_{(q, p)}=\|X\|_{\mathcal{S}_q[\mathcal{S}_p]}$.
		\end{itemize}
	\end{theorem}
	
	We note that it is not immediate that $\|\cdot\|_{(q, p)}$ defined by~\eqref{eq:(q,p)-norm} satisfies the triangle inequality. By the above theorem, however,  $\|\cdot\|_{(q, p)}$ is really a norm and satisfies the triangle inequality.

	\begin{proof}
		(i) We first show that $\|X\|_{\mathcal{S}_p[\mathcal{S}_q]}\leq \|X\|_{(p,q)}$. Suppose that $X\in \mathcal{S}_p[\mathcal{S}_q]$, $Y\in \mathcal{S}_q(\ell_2\otimes_2\ell_2)$ and $A,B\in \mathcal{S}_{2r}$ are such that $X=A Y B$. Suppose also that $Y=C Z D$ is an arbitrary factorization of $Y$ with $C,D\in \mathcal{S}_{2q}$ and $Z\in \mathcal{S}_\infty[\mathcal{S}_q]$. Then, $X=AC\cdot Z\cdot DB$ and by H\"older's inequality we have $AC,DB\in \mathcal{S}_{2p}$. Hence,
		\begin{align*}
			\|X\|_{\mathcal{S}_p[\mathcal{S}_q]}
			&\leq \|AC\|_{\mathcal{S}_{2p}} \|Z\|_{\mathcal{S}_\infty[\mathcal{S}_q]} \|DB\|_{\mathcal{S}_{2p}} \quad\qquad\qquad\qquad \text{(Theorem \ref{Pisier's formula})}\\
			&\leq \|A\|_{\mathcal{S}_{2r}}\|C\|_{\mathcal{S}_{2q}}\|Z\|_{\mathcal{S}_\infty[\mathcal{S}_q]}\|D\|_{\mathcal{S}_{2q}}\|B\|_{\mathcal{S}_{2r}} \ \ \ \ \quad \ \text{(H\"older's inequality)}.
		\end{align*}
		Taking infimum over all $C,D,Z$ as above, and using  Theorem~\ref{Pisier's formula} gives
		\begin{align*}
			\|X\|_{\mathcal{S}_p[\mathcal{S}_q]}
			&\leq \|A\|_{\mathcal{S}_{2r}}\|Y\|_{\mathcal{S}_q[\mathcal{S}_q]}\|B\|_{\mathcal{S}_{2r}}\\
			&= \|A\|_{\mathcal{S}_{2r}}\|Y\|_{\mathcal{S}_q(\ell_2\otimes_2\ell_2)}\|B\|_{\mathcal{S}_{2r}},
		\end{align*}
		where in the second line we use Proposition~\ref{interpolation}(iii). Finally, taking infimum over the choices $A, B, Y$ as above, we obtain $\|X\|_{\mathcal{S}_p[\mathcal{S}_q]}\leq \|X\|_{(p,q)}$ as desired.
		
		To prove the reverse inequality, let $\epsilon>0$ be arbitrary. Then, there exist a $Y\in \mathcal{S}_\infty[\mathcal{S}_q]$ and  $A,B\in \mathcal{S}_{2p}$ such that $X=A Y B$ and $\|A\|_{\mathcal{S}_{2p}} \|Y\|_{\mathcal{S}_\infty[\mathcal{S}_q]} \|B\|_{\mathcal{S}_{2p}}\leq\|X\|_{\mathcal{S}_p[\mathcal{S}_q]}+\epsilon$. 
		We note by the use of polar decomposition and absorbing unitaries into $Y$ without changing its norm, we may assume that $A, B$ are positive.  
		In this case, $A^{\frac{p}{q}},B^{\frac{p}{q}}\in \mathcal{S}_{2q}$ and $A^{\frac{p}{r}},B^{\frac{p}{r}}\in \mathcal{S}_{2r}$ are well-defined. Hence,
		\begin{align*}
			\|X\|_{(p,q)} & = \|AYB\|_{(p, q)}\\
			&=\|A^{\frac{p}{r}}\cdot(A^{\frac{p}{q}}\cdot Y\cdot B^{\frac{p}{q}})\cdot B^{\frac{p}{r}}\|_{(p,q)}\\
			&\leq \|A^{\frac{p}{r}}\|_{\mathcal{S}_{2r}}\|A^{\frac{p}{q}}\cdot Y\cdot B^{\frac{p}{q}}\|_{\mathcal{S}_q(\ell_2\otimes_2\ell_2)}\|B^{\frac{p}{r}}\|_{\mathcal{S}_{2r}}\\
			&=\|A^{\frac{p}{r}}\|_{\mathcal{S}_{2r}}\|A^{\frac{p}{q}}\cdot Y\cdot B^{\frac{p}{q}}\|_{\mathcal{S}_q[\mathcal{S}_q]}\|B^{\frac{p}{r}}\|_{\mathcal{S}_{2r}} \ \ \ \ \qquad \quad\quad\text{(Proposition \ref{interpolation}(iii))}\\
			&\leq \|A^{\frac{p}{r}}\|_{\mathcal{S}_{2r}}\|A^{\frac{p}{q}}\|_{\mathcal{S}_{2q}}\|Y\|_{\mathcal{S}_\infty[\mathcal{S}_q]}\|B^{\frac{p}{q}}\|_{\mathcal{S}_{2q}}\|B^{\frac{p}{r}}\|_{\mathcal{S}_{2r}} \ \ \ \ \ \ \text{(Theorem \ref{Pisier's formula})}\\
			& = \|A\|_{\mathcal{S}_{2p}}^{\frac pr} \|A\|_{\mathcal{S}_{2p}}^{\frac pq} \|Y\|_{\mathcal{S}_\infty[\mathcal{S}_q]}\|B\|_{\mathcal{S}_{2p}}^{\frac pq}\|B\|_{\mathcal{S}_{2p}}^{\frac pr}\\
			&=\|A\|_{\mathcal{S}_{2p}}\|Y\|_{\mathcal{S}_\infty[\mathcal{S}_q]}\|B\|_{\mathcal{S}_{2p}}\\
			&\leq \|X\|_{\mathcal{S}_p[\mathcal{S}_q]}+\epsilon.
		\end{align*}
		Since $\epsilon>0$ is arbitrary, it follows that $\|X\|_{(p,q)}\leq \|X\|_{\mathcal{S}_p[\mathcal{S}_q]}$, and the proof of (i) is complete.
		
		\medskip
		\noindent
		(ii) We prove this part by exploiting the duality relation $\mathcal{S}_q[\mathcal{S}_p]=\mathcal{S}_{q'}[\mathcal{S}_{p'}]^*$ given in Proposition~\ref{interpolation}(ii). Using this fact along with part (i) for $q'\leq p'$, we compute:
		\begin{align*}
			\|X\|_{\mathcal{S}_q[\mathcal{S}_p]}
			&=\|X\|_{\mathcal{S}_{q'}[\mathcal{S}_{p'}]^*}\\
			&=\sup_{Y} \frac{|\tr(Y^* X)|}{\|Y\|_{\mathcal{S}_{q'}[\mathcal{S}_{p'}]}}\\
			&=\sup_{Y} \frac{|\tr(Y^* X)|}{\|Y\|_{(q',p')}}\\
			&=\sup_{Y} \ \sup_{\stackrel{A,B, Z}{Y=AZB}} \frac{|\tr(Y^* X)|}{\|A\|_{\mathcal{S}_{2r}}\|B\|_{\mathcal{S}_{2r}}\|Z\|_{\mathcal{S}_{p'}}}\\
			&=\sup_{A,B, Z}  \frac{\big|\tr(B^* Z^* A^*X)\big|}{\|A\|_{\mathcal{S}_{2r}}\|B\|_{\mathcal{S}_{2r}}\|Z\|_{\mathcal{S}_{p'}}}\\
			&=\sup_{A,B, Z}  \frac{\big|\tr\big(Z^*(A^* X B^*)\big)\big|}{\|A\|_{\mathcal{S}_{2r}}\|B\|_{\mathcal{S}_{2r}}\|Z\|_{\mathcal{S}_{p'}}}\\
			&=\sup_{A,B} \ \frac{\|A^* X B^*\|_{\mathcal{S}_p}}{\|A\|_{\mathcal{S}_{2r}}\|B\|_{\mathcal{S}_{2r}}}\\
			&=\sup_{A,B} \ \frac{\|A X B\|_{\mathcal{S}_p}}{\|A\|_{\mathcal{S}_{2r}}\|B\|_{\mathcal{S}_{2r}}}\\
			&=\|X\|_{(q,p)},
		\end{align*}
	as desired.
	\end{proof}

We conclude this section with the following technical remark which will come in handy later on. For the sake of completeness we include the proofs in Appendix~\ref{proofs}.

	\begin{remark}\label{lem:PSD-case}
		(i) Suppose that $X\in \mathcal{S}_p[\mathcal{S}_q]$ is self-adjoint, i.e., $X^*=X$. Then,
		\begin{equation*}\label{eq:Pisier-formula-00}
			\|X\|_{\mathcal{S}_p[\mathcal{S}_q]}=\inf\|A\|_{\mathcal{S}_{2p}}^2\|Y\|_{\mathcal{S}_\infty[\mathcal{S}_q]},
		\end{equation*}
		where the infimum is taken over all $A\in \mathcal{S}_{2p}$ and $Y\in \mathcal{S}_\infty[\mathcal{S}_q]$ such that $X=A Y A^*$.
		
		(ii) When $X$ is positive, the optimizations in~\eqref{eq:(p,q)-norm} and~\eqref{eq:(q,p)-norm} can be restricted to the case that $B=A\geq0$.
	\end{remark}

	\section{Applications in quantum information theory}\label{sec:applications}

	In this section we present some applications of the above tools in quantum information theory. Here, we denote Hilbert spaces (mostly finite-dimensional) by $\cH_X, \cH_Y$ where subscripts $X, Y$ indicate quantum subsystems. We add a subscript to an operator $T_X\in B(\cH_X)$ to emphasize that it acts on subsystem $X$, and $I_X\in B(\cH_X)$ is the identity operator. A density matrix (quantum state) is a positive operator with normalized trace. For a density matrix $\rho_{XY}\in B(\cH_X\otimes \cH_Y)$, its marginal states $\rho_X\in B(\cH_X)$ and $\rho_Y\in B(\cH_Y)$ are defined via partial trace. Quantum operations (super-operators) are completely positive trace-preserving (cptp) maps $\Phi:B(\cH_X)\to B(\cH_Y)$. We denote the identity super-operator acting on subsystem $X$ by $\mathcal I_X$.

	Recall that for a distribution $p_X$ over discrete set $\mathcal X$, its (Shannon) entropy is defined by
	$$H(X) = -\sum_x p(x)\log p(x).$$
	In the quantum case, for a density matrix $\rho_X$ acting on the finite-dimensional Hilbert space $\cH_X$, its (von Neumann) entropy is defined by
	$$H(X) = - \tr(\rho_X\log \rho_X).$$
	For certain applications it is advantageous to define other forms of entropies that somehow approximate Shannon or von Neumann entropies. A well-known such quantity is the R\'enyi entropy that for a parameter $\alpha$ is given by
	$$H_\alpha(X) = \frac{1}{1-\alpha}\log \tr(\rho_X^\alpha),$$
	which in the classical case (when $\rho_X$ is diagonal with values $p(x), x\in \mathcal X$, on the diagonal) turns into $H_\alpha(X) = \frac{1}{1-\alpha}\log \big(\sum_x p(x)^\alpha\big)$. R\'enyi entropy is an approximation of the usual entropy since 
	$$\lim_{\alpha\to 1}H_\alpha(X)= H(X).$$
	Now, $L_p$-spaces naturally appear in the study of entropic quantities since R\'enyi entropy can be stated in terms of the norm in these spaces:
	$$H_\alpha(X) = \frac{\alpha}{1-\alpha} \log \|\rho_X\|_\alpha.$$
	Therefore, the properties of $L_p$-spaces are applicable in the understanding of these entropic quantities. 
	We emphasize that although R\'enyi quantities can also be defined for $\alpha<1$, hereafter we assume that $\alpha\geq 1$ so that we can use the fact that $\|\cdot\|_\alpha$ is really a norm.

	\subsection{Quantum conditional R\'enyi entropy}
	
	Given a bipartite distribution $p_{XY}$, the condition entropy $H(X|Y)$ is defined by $H(X|Y) = \sum_y p(y)H(X|Y=y)$ as the average entropy of the conditional distribution $p_{X|Y=y}$. Generalization of this definition to the quantum case is tricky since conditional states do not have a well-defined interpretation. Nevertheless, it can be verified that 
	\begin{align}\label{eq:entropy-chain-0}
		H(X|Y) = H(XY) - H(Y),
	\end{align}
	and this formula is readily generalized to the quantum case as the definition of quantum conditional entropy.  The definition of conditional R\'enyi entropy in the quantum case is even more tricky since~\eqref{eq:entropy-chain-0}, called the \emph{chain rule}, no longer holds for $\alpha\neq 1$. Thus, we should take a different approach to define the quantum conditional R\'enyi entropies.

	Umegaki's divergence (also called quantum relative entropy) is another important quantity in information theory which can be considered as a parent quantity for entropy. For a density matrix $\rho_X$, and positive operator $\sigma_X$, that is not necessarily normalized, we define their Umegaki's divergence by
	$$
	D(\rho_X\|\sigma_X):=
	\begin{cases}
		\tr\rho_X(\log\rho_X-\log \sigma_X), \ \ \ \ \text{if} \ \ \supp\rho_X\leq \supp\sigma_X, \\[.1in]
		+\infty, ~~~~~\qquad\qquad\qquad\qquad \ \ \ \text{otherwise}.
	\end{cases}
	$$
	Considering the case where $\rho_X, \sigma_X$ are diagonal, divergence in the classical case is also derived. Moreover, letting $\sigma_X=I_X$ to be the identity operator, we find that 
	$D(\rho_X\|I_X) = -H(\rho_X)$. The conditional entropy can also be stated in terms of divergence:
	$$H(X|Y) = -D(\rho_{XY}\| I_X\otimes \rho_Y).$$
	Even more, we have
	\begin{align}\label{eq:cond-entropy-divergence}
		H(X|Y) = -\inf_{\sigma_Y}D(\rho_{XY}\| I_X\otimes \sigma_Y),
	\end{align}
	where the infimum is taken over all density matrices $\sigma_Y$. This identity can be proven using $D(\rho_{XY}\| I_X\otimes \sigma_Y)=D(\rho_{XY}\| I_X\otimes \rho_Y)+D(\rho_{Y}\| \sigma_Y)$ that is easy to verify, and the fact that $D(\rho_Y\|\sigma_Y)\geq 0$ for any two density matrices $\rho_Y, \sigma_Y$. Equation~\eqref{eq:cond-entropy-divergence} suggests that having a good definition of quantum R\'enyi divergence, we can define a conditional version of quantum R\'enyi entropy.

	R\'enyi divergence in the classical case is given by
	\begin{align*}
		D_\alpha(p_X\| q_X) &= \frac{1}{\alpha-1}\log \sum_x p(x)^{\alpha}q(x)^{1-\alpha}\\
		&= \frac{\alpha}{\alpha-1}\log \Big\|  q_X^{\frac{1-\alpha}{\alpha}}p_X\Big\|_\alpha,
	\end{align*}
	and we have $H_\alpha(p_X)= -D_\alpha(p_X\| I_{X})$ as expected. Next, inspired by~\eqref{eq:cond-entropy-divergence} we may define the conditional R\'enyi entropy by
	\begin{align}\label{eq:cond-Renyi-entropy-c}
		H_\alpha(X|Y) = - \inf_{q_Y} D_\alpha(p_{XY}\| I_A\otimes q_Y).
	\end{align}
	Trying to generalize these definitions to the quantum case a complication arises: since two density matrices $\rho_X, \sigma_X$ do not necessarily commute, $\sigma_X^{\frac{1-\alpha}{\alpha}}$ and $\rho_X$ can be ``multiplied" in different orders. This amounts to different notions of quantum R\'enyi divergence, yet here we limit ourselves to the so-called \emph{sandwiched} R\'enyi divergence~\cite{WWY14, MDSFT13} defined by  
	$$D_\alpha(\rho_X\| \sigma_X) = \frac{\alpha}{\alpha-1}\log \Big\|  \sigma_X^{\frac{1-\alpha}{2\alpha}}\rho_X\sigma_X^{\frac{1-\alpha}{2\alpha}}\Big\|_\alpha,$$
	Based on this and following~\eqref{eq:cond-Renyi-entropy-c} we define the quantum conditional R\'enyi entropy by
	\begin{align*}
		H_\alpha(X|Y) &= - \inf_{\sigma_Y} D_\alpha(\rho_{XY}\| I_X\otimes \sigma_Y)\\
		&= -\inf_{\sigma_Y} \frac{\alpha}{\alpha-1}\log \Big\|  \Big(I_X\otimes \sigma_Y^{\frac{1-\alpha}{2\alpha}}\Big)\rho_{XY} \Big(I_X\otimes \sigma_Y^{\frac{1-\alpha}{2\alpha}}\Big)\Big\|_\alpha.
	\end{align*}

	For $\alpha\geq 1$
	$$\alpha'=\frac{\alpha}{\alpha-1}$$
	is the H\"older conjugate of $\alpha$, and we have  
	\begin{align*}
		H_\alpha(X|Y) &=  -\alpha' \log \Big(\inf_{\sigma_Y}  \Big\|  \Big(I_X\otimes \sigma_Y^{-\frac{1}{2\alpha'}}\Big)\rho_{XY}\Big(I_X\otimes \sigma_Y^{-\frac{1}{2\alpha'}}\Big)\Big\|_\alpha\Big) \\
		&= -\alpha' \log \Bigg(\inf_{\stackrel{T_Y\geq 0}{\|T_Y\|_{\alpha'}=1}}  \Big\|  \Big(I_X\otimes T_Y^{-\frac12}\Big)\rho_{XY}\Big(I_X\otimes T_Y^{-\frac12}\Big)\Big\|_\alpha\Bigg).
	\end{align*}
	Then, comparing this equation with~\eqref{eq:(p,q)-norm} and using Remark~\ref{lem:PSD-case}(ii) we find that
	\begin{align}\label{eq:entropy-1-alpha-norm}
		H_\alpha(X|Y) = -\alpha' \log \big\|\rho_{YX}\big\|_{(1, \alpha)}.
	\end{align}
	This is a key identity that connects operator-valued Schatten spaces to information theoretic quantities.

	\begin{prop}\label{prop:cond-renyi-properties}
		The followings hold for any bipartite density matrix $\rho_{XY}$:
		\begin{enumerate}
			\item[{\rm (i)}] $\lim_{\alpha\rightarrow 1^+} H_\alpha(X|Y) = H(X|Y)$. 
			\item[{\rm (ii)}] $\alpha\mapsto H_\alpha(X|Y)$ is monotonic non-increasing for $\alpha\geq 1$.
			\item[{\rm (iii)}] $H_{\alpha}(X|Y)$ satisfies data processing inequality. More precisely, for any completely positive trace preserving (cptp) map $\Phi: B(\cH_Y)\to B(\cH_{Y'})$ we have $H_{\alpha}(X|Y)_{\rho} \leq H_\alpha(X|Y')_\sigma$ where $\sigma_{XY'} = \mathcal I_A\otimes \Phi(\rho_{XY})$. 
			\item[{\rm (iv)}] $|H_\alpha(X|Y)|\leq \log d_X$ where $d_X=\dim \cH_X$.
		\end{enumerate}
	\end{prop}
	
	\begin{proof} 
		(i) A straightforward computation shows that it suffices to prove that 
		$$\frac{\dd }{\dd \alpha} \|\rho_{YX}\|_{(1, \alpha)}^{\alpha} \bigg|_{\alpha=1} =  H(Y)-H(XY),$$
		which itself is proven in~\cite[Theorem 17]{DJKR}. Interestingly, this latter fact was known even before the definition of quantum conditional R\'enyi entropy. 
		
		\medskip
		\noindent
		(ii) If $\alpha\geq \beta\geq 1$ there is $\theta\in [0,1]$ such that $1/\beta= (1-\theta) + \theta/\alpha$. Therefore, by Remark \ref{rem:interpolation-ineq} we have
		$$\|\rho_{YX}\|_{(1, \beta)}\leq \|\rho_{YX}\|_{(1, 1)}^{1-\theta} \cdot \|\rho_{YX}\|_{(1, \alpha)}^\theta= \|\rho_{YX}\|_{(1, \alpha)}^\theta,$$
		where we used $\|\rho_{YX}\|_{(1, 1)}= \|\rho_{YX}\|_{1}=1$ (see Proposition~\ref{interpolation}(iii)).
		Then, taking the logarithm of both sides and multiplying by $-\beta'$ (where $\beta'$ is the H\"older conjugate of $\beta$), we obtain $H_\beta(X|Y)\geq H_\alpha(X|Y)$.
		
		\medskip
		\noindent 
		(iii) Using~\eqref{eq:entropy-1-alpha-norm} it suffices to show that 
		$$\|\Phi\otimes \mathcal I_X\|_{(1, \alpha)\rightarrow (1, \alpha)}\leq 1.$$
		This holds since using~\cite[Lemma 5]{DJKR} we have $\|\Phi\otimes \mathcal I_X\|_{(1, \alpha)\rightarrow (1, \alpha)} = \|\Phi\|_{1\rightarrow 1}$ which equals $1$ because $\Phi$ is completely positive and trace preserving.
		
		\medskip
		\noindent 
		(iv) Let $\cH_{X'}$ be a Hilbert space with $\dim \cH_{X'}$ equal to the rank of $\rho_X$ that is at most $d_X$. Let $\sigma_{XX'}$ be a \emph{purification} of $\rho_{X}$. Then, there is a cptp map $\Phi: B(\cH_{X'})\to B(\cH_{Y})$ (composition of an isometry and a partial trace) such that $\mathcal I_X\otimes \Phi(\sigma_{XX'})=\rho_{XY}$.\footnote{Here is a construction for $\Phi$: let $\rho_{XYZ}$ be a purification of $\rho_{XY}$, which is a purification of $\rho_{X}$ as well. Then, by the equivalence of purifications~\cite[Exercise 2.81]{NC2010} we have $\dim \cH_{X'}\leq \dim \cH_{Y}\otimes \cH_{Z}$ and  there is an isometry $V_{X'\to YZ}$ such that 
$(I\otimes V) \sigma_{XX'} (I\otimes V^\dagger) = \rho_{XYZ}$. Now, define $\Phi_{X'\to Y}$ by $\Phi (\tau_{X'}) =  \tr_{Z} (V \tau_{X'} V^\dagger)$. From the definition we clearly have $\mathcal I\otimes \Phi(\sigma_{XX'}) = \rho_{XY}$.} 
		Then, by parts (ii) and (iii) we have
		$$H_\infty(X|X')\leq H_\alpha(X|X')\leq H_{\alpha}(X|Y).$$
		Next, by definition we have
		\begin{align*}
			H_\infty(X|X') & = - \log \Bigg(\inf_{\stackrel{T_{X'}\geq 0}{\|T_{X'}\|_{1}=1}}  \Big\|  \Big(I_X\otimes T_{X'}^{-\frac12}\Big)\sigma_{XX'}\Big(I_X\otimes T_{X'}^{-\frac12}\Big)\Big\|_\infty\Bigg)\\
			& = - \log   \Big\|  \Big(I_X\otimes \sigma_{X'}^{-\frac12}\Big)\sigma_{XX'}\Big(I_X\otimes \sigma_{X'}^{-\frac12}\Big)\Big\|_\infty\\
			&\geq - \log  \tr  \Big(\Big(I_X\otimes \sigma_{X'}^{-\frac{1}{2}}\Big)\sigma_{XX'}\Big(I_X\otimes \sigma_{X'}^{-\frac{1}{2}}\Big)\Big)\\
			&\geq - \log d_{X}.
		\end{align*}
		Therefore, $H_{\alpha}(X|Y)\geq -\log d_X$. 
		
		We note that for any unitary $U_X$ we have $\|\rho_{YX}\|_{(1, \alpha)} = \|(I_Y\otimes U_X)\rho_{YX}(I_Y\otimes U_X^*)\|_{(1, \alpha)}$. Therefore, choosing $U_X$ randomly according to the Haar measure, and using the convexity of the $(1, \alpha)$-norm, we find that 
		\begin{align*}
			\|\rho_{YX}\|_{(1, \alpha)} & = \mathbb E \|(I_Y\otimes U_X)\rho_{YX}(I_Y\otimes U_X^*)\|_{(1, \alpha)}\\
			&\geq \|\mathbb E (I_Y\otimes U_X)\rho_{YX}(I_Y\otimes U_X^*)\|_{(1, \alpha)}\\
			&= \Big\|\rho_{Y}\otimes \frac{I_X}{d_X}\Big\|_{(1, \alpha)} \\
			& = \|\rho_Y\|_1\cdot \Big\| \frac{I_X}{d_X}\Big\|_{\alpha}\\
			&= d_X^{-\frac{1}{\alpha'}}.
		\end{align*}
		Using this in the definition of $H_\alpha(X|Y)$, we arrive at $H_\alpha(X|Y)\leq \log d_X$.
		
	\end{proof}
	
	It is well-known that the conditional  von Neumann entropy satisfies a \emph{uniform} continuity bound~\cite{AlickiFannes, Winter16}, i.e., for two quantum states $\rho_{XY}, \sigma_{XY}$ satisfying $\frac 12\|\rho_{XY}-\sigma_{XY}\|_1\leq \epsilon$ we have
	$$\big|H(X|Y)_\rho - H(X|Y)_\sigma\big|\leq 2\epsilon \log d_X +(1+\epsilon)h\Big(\frac{\epsilon}{1+\epsilon}\Big),$$
	where $h(x) = -x\log x-(1-x)\log (1-x)$ is the binary entropy function. The crucial point in this continuity bound is that it is independent of the dimension of $\cH_Y$. Such a continuity bound for the quantum conditional R\'enyi entropy is first found in~\cite{MarwahDupuis} (see also~\cite{JabbourDatta, LWD16} for related results). Here, we prove another such continuity bound that is stronger than the result of~\cite{MarwahDupuis} in some regimes of parameters.

	\begin{theorem}\label{thm:unif-contin-bound}
		Let $\rho_{XY}, \sigma_{XY}$ be two quantum states satisfying $\frac 12\|\rho_{XY}-\sigma_{XY}\|_1\leq \epsilon$. Then we have
		$$\big|H_\alpha(X|Y)_\rho - H_\alpha(X|Y)_\sigma\big|\leq \alpha'\log(1+2\epsilon d_X^{2/\alpha'}) .$$
	\end{theorem}
	
	We note that our bound diverges as $\alpha\to 1$, yet it is stronger than the bound of~\cite{MarwahDupuis} for large values of $\alpha$.

	\begin{proof}
		Let us assume with no loss of generality that  $\|\rho_{XY}-\sigma_{XY}\|_1= 2\epsilon$. Then, there are density matrices $\mu_{XY}, \nu_{XY}$ such that 
		$$\rho_{XY}-\sigma_{XY} = \epsilon(\mu_{XY}- \nu_{XY}).$$
		Therefore,
		\begin{align*}
			\Big|\|\rho_{YX}\|_{(1,\alpha)}-\|\sigma_{YX}\|_{(1,\alpha)}\Big| &\leq \|\rho_{YX}-\sigma_{YX}\|_{(1, \alpha)}\\
			&=\epsilon\|\mu_{YX}- \nu_{YX}\|_{(1, \alpha)}\\
			&\leq \epsilon\big(\|\mu_{YX}\|_{(1, \alpha)} +\| \nu_{YX}\|_{(1, \alpha)} \big).
		\end{align*}
		On the other hand, as a consequence of part (iv) of Proposition~\ref{prop:cond-renyi-properties} we have 
		$$\|\mu_{YX}\|_{(1, \alpha)}, \| \nu_{YX}\|_{(1, \alpha)}\leq d_X^{1/\alpha'}.$$
		Hence, 
		$$\Big|\|\rho_{YX}\|_{(1,\alpha)}-\|\sigma_{YX}\|_{(1,\alpha)}\Big|\leq 2\epsilon d_X^{1/\alpha'}.$$
		As a result,
		\begin{align*}
			\frac{\|\rho_{YX}\|_{(1, \alpha)}}{   \|\sigma_{YX}\|_{(1, \alpha)}  } &\leq \frac{\|\sigma_{YX}\|_{(1, \alpha)} + 2\epsilon d_X^{1/\alpha'} }{   \|\sigma_{YX}\|_{(1, \alpha)}  } \\
			&= 1+ \frac{ 2\epsilon d_X^{1/\alpha'} }{   \|\sigma_{YX}\|_{(1, \alpha)}  }\\
			&\leq 1+2\epsilon d_X^{2/\alpha'},
		\end{align*}
		where in the last inequality we once again use part (iv) of Proposition~\ref{prop:cond-renyi-properties}. Taking the logarithm of both sides and multiplying by $\alpha'$ yields
		$$H_\alpha(X|Y)_\sigma- H_\alpha(X|Y)_\rho\leq \alpha'\log(1+2\epsilon d_X^{2/\alpha'}).$$
		
	\end{proof}

	The following theorem is a generalization of the chain rule for the von Neumann entropy.

	\begin{theorem}
		Let $\alpha, \beta, \gamma\geq 1$ such that
		$$
		\frac{\alpha}{\alpha-1}=\frac{\beta}{\beta-1}+\frac{\gamma}{\gamma-1}.
		$$
		Then for any tripartite density matrix $\rho_{XYZ}$ we have
		\begin{equation*}\label{eq:chain-rule-Renyi}
			H_\alpha(XY|Z) \geq H_\beta(X|YZ)+H_\gamma(Y|Z).
		\end{equation*}
	\end{theorem}
	
	This theorem is first proven in~\cite{Dupuis}. The following proof, that is much simpler than the original proof in~\cite{Dupuis}, is found by the first author of this paper and independently by McKinlay and Tomamichel (private communications). See also~\cite{McKinlayTomamichel} for other results in this direction. 
	
	\begin{proof}
		Note first that the condition on $\alpha, \beta, \gamma$ may be rewritten as ${\alpha'}= {\beta'} + {\gamma'}$. Then, for $\theta:=\gamma'/\alpha'=1-\beta'/\alpha'$ we have 
		$\theta\in [0,1]$ and 
		$$\frac{1}{\alpha} = \frac{1-\theta}{1} + \frac{\theta}{\gamma},  \qquad \frac{1}{\alpha} = \frac{1-\theta}{\beta} + \frac{\theta}{1}.$$
		Now by invoking Proposition \ref{interpolation}(i) we have
		$$
		 \mathcal{S}_1 \big[\mathcal{S}_\alpha [\mathcal{S}_\alpha]\big]=\Big[\mathcal{S}_1\big[ \mathcal{S}_1[ \mathcal{S}_\beta]\big],\mathcal{S}_1\big[ \mathcal{S}_\gamma[\mathcal{S}_1]\big]\Big]_\theta.
		$$
		That is, the $(1, \alpha, \alpha)$-norm is obtained via interpolation between the $(1, 1, \beta)$- and $(1, \gamma, 1)$-norms. 
		Therefore, by Remark \ref{rem:interpolation-ineq} we have
		$$\|\rho_{ZYX}\|_{(1, \alpha, \alpha)} \leq \|\rho_{ZYX}\|_{(1, 1, \beta)}^{1-\theta} \cdot \|\rho_{ZYX}\|_{(1, \gamma, 1)}^\theta.$$
		On the other hand, by Proposition~\ref{interpolation}(iii) we have
		$$\|\rho_{ZYX}\|_{(1, \alpha, \alpha)} = \|\rho_{Z,YX}\|_{(1, \alpha)},$$
		and
		$$\|\rho_{ZYX}\|_{(1, 1, \beta)}=\|\rho_{ZY, X}\|_{(1, \beta)}.$$
		Furthermore, as shown in~\cite{DJKR} we have
		$$\|\rho_{ZYX}\|_{(1, \gamma, 1)} = \|\rho_{ZY}\|_{(1, \gamma)}.$$
		Putting all these together we arrive at
		$$\|\rho_{Z, YX}\|_{(1, \alpha)} \leq \|\rho_{ZY, X}\|_{(1, \beta)}^{1-\theta} \cdot \|\rho_{ZY}\|_{(1, \gamma)}^\theta.$$
		Now taking the logarithm of both sides and multiplying by $-\alpha'$, the desired inequality is obtained.
	\end{proof}

	\subsection{R\'enyi coherent information}
	
	The coherent information $I(X\rangle Y)$ of a state $\rho_{XY}$ is defined by 
	$$I(X\rangle Y) = H(Y) - H(XY) = -H(X|Y).$$
	Then, given a cptp map $\Phi:B(\cH_X)\to B(\cH_Y)$, its coherent information is defined by
	\begin{align}\label{eq:coy-channel}
		I^{\coh}(\Phi) = \max_{\rho_{RX}} I(R\rangle Y)_\sigma,
	\end{align}
	where $I(R\rangle Y)_\sigma$ is computed with respect to 
	$$\sigma_{RY}= \mathcal I_R\otimes \Phi(\rho_{RX}),$$
	and $\cH_R$ is an auxiliary reference Hilbert space. 
	It is well-known that $Q(\Phi)$, the quantum capacity of $\Phi$, i.e., its optimal rate for sending quantum information with asymptotically vanishing error is equal to~\cite{WildeBook} 
	$$Q(\Phi) = \lim_{n\rightarrow\infty} \frac{1}{n} I^{\coh}(\Phi^{\otimes n}),$$
	where $\Phi^{\otimes n}$ is the $n$-fold tensor product of $\Phi$. 
	
	Since coherent information is expressed in terms of conditional entropy, we may define R\'enyi coherent information based on the R\'enyi conditional entropy.
	Indeed, based on~\eqref{eq:coy-channel} we may define the $\alpha$-R\'enyi coherent information of a cptp map by~\cite{GaoJungeLaRacuente}
	\begin{align*}
		I^{\coh}_{\alpha}(\Phi) &= \max_{\rho_{RX}} I_\alpha(R\rangle Y)_\sigma\\
		&= \alpha' \log\Big( \max_{\rho_{RX}} \| \Phi\otimes \mathcal I_R (\rho_{XR})  \|_{(1, \alpha)}    \Big).
	\end{align*}
	$$$$
	
	\begin{prop}\label{prop:coherent-inf}
		$I^{{\rm coh}}_{\alpha}(\Phi) = \alpha'\log \|\Phi\otimes \mathcal I_R\|_{(1, 1)\rightarrow (1, \alpha)}$.
	\end{prop}
	
	\begin{proof}
		We need to show that 
		$$\|\Phi\otimes \mathcal I_R\|_{(1, 1)\rightarrow (1, \alpha)} = \max_{\rho_{RA}} \| \Phi\otimes \mathcal I_R (\rho_{XR})  \|_{(1, \alpha)},$$
		where the maximum is taken over density matrices $\rho_{XR}$. Since $\|\rho_{XR}\|_{(1, 1)} = \|\rho_{XR}\|_1=1$ for any density matrix $\rho_{XR}$, it suffices to show that the maximum in
		$$\max_{F_{XR}: \|F\|_1=1} \| \Phi\otimes \mathcal I_R (F_{XR})  \|_{(1, \alpha)},$$
		is achieved at a positive $F$. To prove this we follow similar steps as in the proof of~\cite[Theorem 12]{DJKR}.
		
		Let $F_{XR}$ be such that $\|F\|_1=1$. Then, by the polar decomposition there exists a unitary $U_{XR}$ and density matrix $\rho_{XR}$ such that $F = U\rho$. Now, the block matrix 
		$$\begin{bmatrix}
			U\rho^{1/2}\\ \\\rho^{1/2}
		\end{bmatrix}  \begin{bmatrix}
			\rho^{1/2} U^* &  \rho^{1/2}
		\end{bmatrix} = \begin{bmatrix}
			U\rho U^*  & U \rho\\
			\rho U^* & \rho
		\end{bmatrix},$$
		is positive. Then, since $\Phi$ is completely positive, 
		$$\begin{bmatrix}
			\Phi\otimes \mathcal I_R(U\rho U^*)  & \Phi\otimes \mathcal I_R(U \rho)\\
			\Phi\otimes \mathcal I_R(\rho U^*) & \Phi\otimes \mathcal I_R(\rho)
		\end{bmatrix},$$
		is positive. 
		Therefore, using~\cite[Proposition 1.3.2]{Bhatia-P} there exists a contraction  $K_{YR}$ such that
		$$\Phi\otimes \mathcal I_R(U\rho) = \big(\Phi\otimes \mathcal I_R(U\rho U^*)\big)^{1/2}K_{YR}\big(\Phi\otimes \mathcal I_R(\rho)    \big)^{1/2}.$$
		Hence, according to~\cite[Lemma 9]{DJKR} we have 
		\begin{align*}
			\|\Phi\otimes \mathcal I_R(F)\|_{(1, \alpha)} &\leq \|\Phi\otimes \mathcal I_R(U\rho U^*)\|_{(1, \alpha)}^{1/2}\cdot \|\Phi\otimes \mathcal I_R(\rho)\|_{(1, \alpha)}^{1/2}\\
			& \leq \max\Big\{  \|\Phi\otimes \mathcal I_R(U\rho U^*)\|_{(1, \alpha)}, \|\Phi\otimes \mathcal I_R(\rho)\|_{(1, \alpha)}\Big\}.
		\end{align*}
		The proof concludes once we note that both $\rho$ and $U\rho U^*$ are density matrices.
	\end{proof}
	
	Now that we represent the $\alpha$-R\'enyi coherent information of a channel in terms of an operator norm, we may use all the tools for studying operators norms, such as convexity, interpolation techniques and log-Sobolev inequalities to investigate the problem of quantum channel capacity. In particular, it is tempting to use the $\alpha$-R\'enyi coherent information to prove strong converse bounds on quantum channel capacity. A natural open question in this direction is whether $I_\alpha^{\coh}(\cdot)$ is \emph{additive} for certain (e.g., degradable) cptp maps or not, i.e., does  $I_\alpha^{\coh}(\Phi_1\otimes\Phi_2)$ equal to  $I_\alpha^{\coh}(\Phi_1)+ I_\alpha^{\coh}(\Phi_2)$?

	We finish this subsection by mentioning some related results to the above additivity problem. Motivated by the problem of the additivity of the \emph{minimum output entropy} of cptp maps, it is shown in~\cite{DJKR} that the completely bounded $(p\to q)$-norm of cptp maps is multiplicative. The proof of this result is heavily based on the theory of operator-valued Schatten norms. It is also shown that the \emph{entanglement-assisted} capacity of cptp maps can be expressed in terms of the so-called completely $p$-summing norms~\cite{JungePalazuelos} that is again related to the theory of  operator-valued Schatten norms. Inspired by~\cite{DJKR}, the theory of \emph{logarithmic-Sobolev} inequalities of completely bounded norms is developed in~\cite{BeigiKing}. Also, extending~\cite{BeigiKing}, estimations of the decoherence rate of \emph{non-primitive} quantum Markov semigroups is derived in~\cite{BardetRouze}. See also~\cite{BJLRF21} for some related results.

	\subsection{A correlation measure} 
	In this subsection we define yet another information theoretic invariant based on operator-valued Schatten spaces. Here, we briefly explain the main properties of this invariant that is first introduced in~\cite{MBGYA} and refer to the original paper for more details. 
	
	For a bipartite density matrix $\rho_{XY}$ define
	$$W_\alpha(X|Y)=\Big\| \rho_{YX} - \rho_Y\otimes \frac{I_X}{d_X}  \Big\|_{(1, \alpha)},$$
	where $d_X=\dim \cH_X$ is the dimension of subsystem $X$. Similar to $H_\alpha(X|Y)$, the definition of $W_\alpha(X|Y)$ is also in terms of the $(1, \alpha)$-norm, so it is expected that these two quantities are related. Indeed, it is shown in~\cite[Proposition 9]{MBGYA} that
	$$ \Big|W_\alpha(X|Y) -e^{-\frac{1}{\alpha'}H_\alpha(X|Y)}\Big|\leq  d_X^{-\frac{1}{\alpha'}}.$$

	We note that if $\rho_{XY}=\rho_X\otimes \rho_Y$ is a product density matrix and $\rho_X= I_X/d_X$, then $W_\alpha(X|Y)=0$. Thus, assuming that marginal density matrix $\rho_X$ is maximally mixed, then $W_\alpha(X|Y)$ can be thought of as a correlation measure. As a correlation measure it is expected that $W_\alpha(X|Y)$ satisfies the data processing inequality: if $\Phi: B(\cH_Y)\to B(\cH_Y')$ is a cptp map, then
	$$W_\alpha(X|Y')\leq W_\alpha(X|Y).$$
	This inequality is proven in~\cite[Theorem 8]{MBGYA} based on similar ideas as in the proof of part (iii) of Proposition~\ref{prop:cond-renyi-properties}. 
	
	\begin{theorem}\cite[Corollary 14]{MBGYA} \label{thm:decoupling}
		Let $\rho_{XY}$ be a bipartite density matrix and let $P:\cH_X\to \cH_{X_0}$ be an orthogonal projection where $\cH_{X_0}\subseteq \cH_X$ is a subspace. Then for any $1\leq \alpha\leq 2$ we have
		$$\mathbb E_{U}\Big[  \Big\|  \frac{d_X}{d_{X_0}} (PU_X\otimes I_Y)\rho_{XY} (U^*_{X}P\otimes I_Y) - \frac{I_{X_0}}{d_{X_0}}\otimes \rho_Y    \Big\|_1  \Big]\leq 2^{\frac{2}{\alpha}-1}d_{X_0}^{\frac{1}{\alpha'}}W_\alpha(X|Y).$$
	\end{theorem}
	
	Proof of this theorem for $\alpha=1$ is easy. For $\alpha=2$, it is proven essentially based on ideas from~\cite{DBWR14}. Then, the key observation is that, by Theorem~\ref{thm:interpolation-bound}, from the boundary cases $\alpha\in \{1,2\}$ we can infer the theorem for all values of $1\leq \alpha\leq 2$.  See~\cite{MBGYA} for more details.

	Theorem~\ref{thm:decoupling} is called a decoupling theorem as it says that by applying a random unitary and a projection on a bipartite density matrix, the two parts become decoupled. Such decoupling results have several applications in information theory, e.g., in channel coding and privacy amplification~\cite{DBWR14, MBGYA, Sharma}.

	\section*{Acknowledgements} MMG is supported by a fund from MSRT. SB is thankful to Marco Tomamichel and Alexander McKinlay for discussions that motivated this work. The authors are also grateful to Paul Gondolf for notifying us that our bound in Theorem~\ref{thm:unif-contin-bound} diverges as $\alpha\to 1$. Finally, the authors are thankful to anonymous referees whose comments greatly improved the presentation of the paper.

	\appendix

	\section{Missing proofs from Sections~\ref{sec:operator-spaces} and~\ref{sec:NCVV}}\label{proofs}

	\begin{proof}[Proof of Proposition \ref{prop:h-tensor-R-C}]
		Part (i) is straightforward. 
		We prove (ii) in the special case of $\cH=\ell_2^d$; generalization to arbitrary $\cH$ is similar.  	
		
		We first show that $\mathcal{S}_\infty^d = \cC^d\otimes_h \cR^d$. Recall that $\cC^d$ is the space of column vectors, and $\cR^d$ is the space of row vectors. Then, for $v^c\in \cC^d$ and $w^r\in \cR^d$, their multiplication $v^cw^r$ makes sense and forms a $d\times d$ matrix, i.e., an element of $\mathcal{S}_\infty^d$. Define the linear map $\phi:  \cC^d\otimes_h \cR^d\to \mathcal{S}_\infty^d$ by 
		$$\phi(v^c\otimes w^r) := v^cw^r.$$
		We claim that $\phi$ is a complete isometry between $ \cC^d\otimes_h \cR^d$ and $\mathcal{S}_\infty^d$. To this end, we prove that $\phi=\phi_1$ is an isometry; extending the proof for $\phi_n$, the $n$-th amplification of $\phi$, is immediate. 
		
		Let $Z\in \cC^d\otimes_h\cR^d$. Then,
		\begin{align*}
			\|Z\|_{ \cC^d\otimes_h\cR^d} & = \inf\big\{   \|X\|_{M_{1, m}(\cC^d)} \|Y\|_{M_{m, 1}(\cR^d)}  :\, X\in M_{1, m}(\cC^d) , Y\in M_{m, 1}(\cR^d), Z=X\odot Y \big\}.
		\end{align*}
		Put 
		$$X = [v_1^c \dots v_m^c]\quad \text{ and }\quad 
		Y=\begin{bmatrix}
			w_1^r\\
			\vdots\\
			w_m^r
		\end{bmatrix}.$$
		Then,  
		$$\|X\|_{M_{1, m}(\cC^d)} = \big\|[v_1^c \dots v_m^c]\big\|_{d, m}\quad \text{ and }\quad 
		\|Y\|_{M_{m, 1}(\cR^d)}=\left\|\begin{bmatrix}
			w_1^r\\
			\vdots\\
			w_m^r
		\end{bmatrix}\right\|_{m, d}.$$
		On the other hand, $Z=X\odot Y$ corresponds to 
		$\phi(Z) = \sum_i v^c_i w^r_i $.
		Therefore, $\|Z\|_{ \cC^d\otimes_h\cR^d}$ equals the infimum of 
		$$\big\|[v_1^c \dots v_m^c]\big\|_{d, m} 
		\left\|\begin{bmatrix}
			w^r_1\\
			\vdots\\
			w_m^r
		\end{bmatrix}\right\|_{m, d},$$
		subject to $\phi(Z) = \sum_i v^c_i w^r_i $. Note that this is nothing but the operator norm of $\phi(Z)$, i.e., $\|\phi(Z)\|_{\mathcal{S}_\infty^d}$.

		Next, we show that $\mathcal{S}_1^d = \cR^d\otimes_h \cC^d$. As before, thinking of $\cR^d$ and $\cC^d$ as row and column vectors, we can define the linear map $\psi:  \cR^d\otimes_h \cC^d\to \mathcal{S}_1^d$ by 
		$$\psi(v^r\otimes w^c):= w^cv^r.$$ 
		Again we only show that $\psi$ is an isometry. To see this, let $Z\in \cR^d\otimes_h \cC^d$. Then,
		$$
		\|Z\|_{\cR^d\otimes_h \cC^d} = \inf \{ \|X\|_{M_{1, m}(\cR^d)} \|Y\|_{M_{m, 1}(\cC^d)}  :\, X\in M_{1, m}(\cR^d) , Y\in M_{m, 1}(\cC^d), Z=X\odot Y\}.
		$$
		Put
		$$X = [v_1^r \dots v_m^r]\quad \text{ and }\quad 
		Y=\begin{bmatrix}
			w_1^c\\
			\vdots\\
			w_m^c
		\end{bmatrix}.$$
		Then,  
		\begin{align}\label{eq:M-R-d}
			\|X\|_{M_{1, m}(\cR^d)} = \big\|[v_1^r \dots v_m^r]\big\|_{1, md} = \Big(\sum_i \|v^r_i\|^2\Big)^{1/2},
		\end{align}
		and
		\begin{align}\label{eq:M-C-d} 
			\|Y\|_{M_{m, 1}(\cC^d)}=\left\|\begin{bmatrix}
				w^c_1\\
				\vdots\\
				w^c_m
			\end{bmatrix}\right\|_{md, 1}=\Big(\sum_i \|w^c_i\|^2\Big)^{1/2}.
		\end{align}
		On the other hand, $Z=X\odot Y$ corresponds to 
		$\psi(Z) = \sum_i w^c_iv^r_i$.
		Therefore, $\|Z\|_{ \cR^d\otimes_h\cC^d}$ equals the infimum of 
		$$\Big(\sum_i \|v^r_i\|^2\Big)^{1/2}\Big(\sum_i \|w^c_i\|^2\Big)^{1/2},$$
		subject to $\psi(Z) = \sum_i w^c_iv^r_i$. An application of H\"older's inequality shows that this is equal to  $\|\psi(Z)\|_{\mathcal{S}_1^d}$.  
		
		Part (iii) follows readily once we note that $M_{1, d}(\cC^d)$ and $M_{d, 1}(\cR^d)$ are naturally identified with the space of $d\times d$ matrices,  and by definition, their norms are the same as the norm of $\mathcal{S}_\infty^d$. The isometries $M_{1, d}(\cR^d) = M_{d, 1}(\cC^d) = \mathcal{S}_2^d$ are proven essentially in~\eqref{eq:M-R-d} and~\eqref{eq:M-C-d}.
\end{proof}
	
	\begin{proof}[Proof of Proposition \ref{interpolation}]
		(i) First note that by using the reiteration property (Proposition \ref{operatorspaceinterpolation}), we observe that $\mathcal{R}\big(1-\tfrac1{p_\theta}\big)=\Big[\mathcal{R}\big(1-\tfrac1{p_0}\big),\mathcal{R}\big(1-\tfrac1{p_1}\big)\Big]_\theta$ and $\mathcal{R}\big(\tfrac1{p_\theta}\big)=\Big[\mathcal{R}\big(\tfrac1{p_0}\big),\mathcal{R}\big(\tfrac1{p_1}\big)\Big]_\theta$. Then,
		\begin{align*}
			\mathcal{S}_{p_\theta}[\mathcal{X}_\theta] &=\mathcal{R}\big(1-\tfrac1{p_\theta}\big)\otimes_h \Big(\mathcal{X}_\theta\otimes_h\mathcal{R}\big(\tfrac1{p_\theta}\big)\Big) \qquad   \text{(Theorem~\ref{identification})}\\
			&=\mathcal{R}\big(1-\tfrac1{p_\theta}\big)\otimes_h\Big[\mathcal{X}_0\otimes_h\mathcal{R}\big(\tfrac1{p_0}\big),\mathcal{X}_1\otimes_h\mathcal{R}\big(\tfrac1{p_1}\big)\Big]_\theta \qquad   \text{(Theorem~\ref{KouPis})}\\
			&=\Big[\mathcal{R}\big(1-\tfrac1{p_0}\big)\otimes_h\mathcal{X}_0\otimes_h\mathcal{R}\big(\tfrac1{p_0}\big),\mathcal{R}\big(1-\tfrac1{p_1}\big)\otimes_h\mathcal{X}_1\otimes_h\mathcal{R}\big(\tfrac1{p_1}\big)\Big]_\theta \qquad   \text{(Theorem~\ref{KouPis})}\\
			&=\big[\mathcal{S}_{p_0}[\mathcal{X}_0], \, \mathcal{S}_{p_1}[\mathcal{X}_1]\big]_\theta \qquad   \text{(Theorem~\ref{identification})}.
		\end{align*}
		With Remark \ref{remark3.5} in mind, \eqref{Pis2} is obtained from \eqref{Pis1} by taking $\mathcal{X}_i=\mathcal{S}_{p_i}$ for $i=0,1$.
		
		(ii) Again, we prove the theorem in the finite-dimensional case, i.e., $\mathcal{S}_p^d[\mathcal{X}]^*=\mathcal{S}^d_{p'}[\mathcal{X}^*]$. We know that for operator spaces $\cY_1, \cY_2$ which \emph{one} of them is finite-dimensional we have $(\cY_1\otimes_h\cY_2)^* = \cY_1^*\otimes_h \cY_2^*$ completely isometrically~\cite[Corollary 5.8]{Pisier3}. Using this result twice we obtain
		\begin{align*}
			\mathcal{S}^d_p[\mathcal{X}]^* & = \Big(\cR^d(1-1/p)\otimes_h \cX \otimes_h \cR^d(1/p)\Big)^*\\
			& = \cR^d(1-1/p)^*\otimes_h \Big(\cX \otimes_h \cR^d(1/p)\Big)^*\\
			& = \cR^d(1-1/p)^*\otimes_h \cX^* \otimes_h \cR^d(1/p)^*,
		\end{align*}
		completely isometrically. Next, we have
		$$\cR^d(1/p)^* = [\cR^d, \cC^d]_{1/p}^* = \big[\big(\cR^d\big)^*, \big(\cC^d\big)^*\big]_{1/p} = \big[\cC^d, \cR^d\big]_{1/p} =\big[\cR^d, \cC^d\big]_{1/p'}= \cR^d(1/p'),$$
		completely isometrically. Here, in the second equality we use the fact that for a compatible couple $(\cY_1, \cY_2)$ of finite-dimensional operator spaces we have $[\cY_1, \cY_2]_\theta^* = [\cY_1^*, \cY_2^*]_\theta$ completely isometrically~\cite[Theorem 2.7.4]{Pisier3}. We also for the third equality use $(\cC^d)^*=\cR^d$ and $(\cR^d)^*=\cC^d$ that are easy to prove~\cite[p.~41]{Pisier3}. Then, we have
		\begin{align*}
			\mathcal{S}^d_p[\mathcal{X}]^* = \cR^d(1/p)\otimes_h \cX^* \otimes_h \cR^d(1/p') = S^d_{p'}[\cX^*],
		\end{align*}
		completely isometrically.
		
		(iii) For $p=\infty$ we have
		\begin{align*}
			\mathcal{S}_\infty[\mathcal{S}_\infty]
			&=\mathcal{C}\otimes_h \mathcal{S}_\infty\otimes_h \mathcal{R}\\
			&=(\mathcal{C}\otimes_h \mathcal{C})\otimes_h(\mathcal{R}\otimes_h \mathcal{R})\\
			&=(\ell_2\otimes_2\ell_2)_c\otimes_h(\ell_2\otimes_2\ell_2)_r \qquad \qquad   \text{(Proposition~\ref{prop:h-tensor-R-C}(i))}\\
			&=\mathcal{S}_\infty(\ell_2\otimes_2\ell_2) \qquad \qquad   \text{(Proposition~\ref{prop:h-tensor-R-C}(ii))}.
		\end{align*}
		The case $p=1$ is similar. The general case is now obtained by interpolation:
		\begin{align*}
			\mathcal{S}_p[\mathcal{S}_p]
			&=\Big[\mathcal{S}_\infty[\mathcal{S}_\infty],\mathcal{S}_1[\mathcal{S}_1]\Big]_{\frac{1}{p}} \qquad \qquad  \text{\eqref{Pis2}}\\
			&=\Big[\mathcal{S}_\infty(\ell_2\otimes_2\ell_2),\mathcal{S}_1(\ell_2\otimes_2\ell_2)\Big]_{\frac{1}{p}}\\
			&=\mathcal{S}_p(\ell_2\otimes_2\ell_2) \qquad \qquad  \text{\eqref{eq:Sp-interpolation}}.
		\end{align*}
		The proof is complete.
	\end{proof}

\begin{proof}[Proof of Remark \ref{lem:PSD-case}]
	(i) 	We need to show that for any $\epsilon>0$ there is a factorization $X=\hat A\hat Y\hat A^*$ 
	such that
	$$\|X\|_{\mathcal{S}_p[\mathcal{S}_q]} +\epsilon \geq \|\hat A\|_{\mathcal{S}_{2p}}\|\hat Y\|_{\mathcal{S}_\infty[\mathcal{S}_q]}\|\hat A^*\|_{\mathcal{S}_{2p}}.$$
	Following Theorem~\ref{Pisier's formula}, assume that $X=AYB$ is such that  $\|A\|_{\mathcal{S}_{2p}}= \|B\|_{\mathcal{S}_{2p}}$ and
	$$\|X\|_{\mathcal{S}_p[\mathcal{S}_q]}+\epsilon\geq \|A\|_{\mathcal{S}_{2p}}\|Y\|_{\mathcal{S}_\infty[\mathcal{S}_q]}\|B\|_{\mathcal{S}_{2p}}.$$
	Define $\hat A :=[A ~ B^*]$ and 
	\begin{equation*}
		\hat Y= \frac{1}{2}\begin{bmatrix}
			0 & Y\\
			Y^* & 0
		\end{bmatrix}.
	\end{equation*}

	Then we have $X=\hat A\hat Y\hat A^*$, and 
	\begin{align*}
		\|\hat A\|_{\mathcal{S}_{2p}}\|\hat Y\|_{\mathcal{S}_\infty[\mathcal{S}_q]}\|\hat A^*\|_{\mathcal{S}_{2p}}& = \frac{1}{2} \|AA^*+B^*B\|_{\mathcal{S}_{p}} \|Y\|_{\mathcal{S}_\infty[\mathcal{S}_q]}\\
		& \leq  \frac{1}{2}\big( \|AA^*\|_{\mathcal{S}_{p}}+\|B^*B\|_{\mathcal{S}_{p}}\big) \|Y\|_{\mathcal{S}_\infty[\mathcal{S}_q]}\\
		& = \|A\|_{\mathcal{S}_{2p}}\|B\|_{\mathcal{S}_{2p}}\|Y\|_{\mathcal{S}_\infty[\mathcal{S}_q]}\\
		&\leq \|X\|_{\mathcal{S}_p[\mathcal{S}_q]}+\epsilon,
	\end{align*}
	where in the third line we used $\|A\|_{\mathcal{S}_{2p}}= \|B\|_{\mathcal{S}_{2p}}$. 
	
	(ii) Following~\cite{DJKR}, when $X$ is positive, by H\"older's inequality we have
	\begin{align*}
		\|A X B\|_{\mathcal{S}_p}& \leq \|AX^{\frac 12}\|_{\mathcal{S}_{2p}}\|X^{\frac 12}B\|_{\mathcal{S}_{2p}}  \\
		&=  \|AXA^*\|_{\mathcal{S}_{p}}^{\frac 12}\|B^*XB\|_{\mathcal{S}_{p}}^{\frac 12}\\
		&\leq \max\big\{\|AXA^*\|_{\mathcal{S}_{p}},\, \|B^*XB\|_{\mathcal{S}_{p}}\big\}.
	\end{align*}
	This shows that the supremum in~\eqref{eq:(q,p)-norm} can be restricted to the case of $B=A^*$.  Moreover, as argued in the proof of Theorem~\ref{thm:Pisier-Junge-formula} we can apply polar decomposition to assume that $A$ is positive.
	
	To prove the assertion for~\eqref{eq:(p,q)-norm} we follow similar steps as in the proof of Theorem~\ref{thm:Pisier-Junge-formula}.  First, by part (i) of the remark, for any $\epsilon>0$ there exists a factorization $X=AYA^*$ such that 
	$$\|A\|_{\mathcal{S}_{2p}}^2\|Y\|_{\mathcal{S}_\infty[\mathcal{S}_q]}\leq \|X\|_{\mathcal{S}_p[\mathcal{S}_q]}+\epsilon.$$
	Moreover, once again applying the polar decomposition we can assume that $A>0$.
	Let $\hat A = A^{\frac pr}$ and $\hat Y = A^{\frac pq}Y A^{\frac pq}$. Then, we have $X= \hat A \hat Y\hat A^*$ and $\|\hat A\|_{\mathcal{S}_{2r}}= \|\hat A^*\|_{\mathcal{S}_{2r}}= \|A\|_{\mathcal{S}_{2p}}^{\frac pr}$.  Moreover, by Theorem \ref{Pisier's formula} we have
	$$
	\|Y\|_{\mathcal{S}_{\infty}[\mathcal{S}_q]}\geq \frac{\|A^{\frac pq}YA^{\frac pq}\|_{\mathcal{S}_q[\mathcal{S}_q]}}{\|A^{\frac pq}\|^2_{\mathcal{S}_{2q}}} = \frac{\|\hat Y\|_{\mathcal{S}_q[\mathcal{S}_q]}}{\|A\|_{\mathcal{S}_{2p}}^{\frac {2p}q}}.
	$$
	Therefore, 
	\begin{align*}
		\|\hat A\|_{\mathcal{S}_{2r}} \|\hat Y\|_{\mathcal{S}_q[\mathcal{S}_q]}\|\hat A^*\|_{\mathcal{S}_{2r}}& \leq \|A\|_{\mathcal{S}_{2p}}^{\frac {2p}r}\cdot \|A\|_{\mathcal{S}_{2p}}^{\frac {2p}q}\cdot \|Y\|_{\mathcal{S}_{\infty}[\mathcal{S}_q]} \\
		&= \|A\|_{\mathcal{S}_{2p}}^{2}\cdot \|Y\|_{\mathcal{S}_{\infty}[\mathcal{S}_q]} \\
		&\leq \|X\|_{\mathcal{S}_p[\mathcal{S}_q]}+\epsilon.
	\end{align*}
	The proof is complete. 
\end{proof}

	\section{Auxiliary lemmas in the proof of Theorem~\ref{KouPis}}\label{sec:Subharmonicity}
	
	\subsection{A selection lemma}
	
	Suppose that $\alpha$, $\beta$ and $\gamma$ are norms, respectively, on $M_{d,n}$, $M_{n,d'}$ and $M_{d,d'}$ such that
	\begin{equation*}
		\gamma(X)=\inf\big\{\alpha(A)\cdot \beta(B):X=A B, A\in M_{d,n}, B\in M_{n,d'}\big\}.
	\end{equation*}
	
	Our goal here is to show that \emph{almost optimizing} pairs $(A,B)$ can be chosen to be measurable functions.
	
	\begin{lemma}\label{selection lemma}
		Let $X(s)$ be a continuous function on $\mathbb{R}$ with values in $M_{d,d'}$. 
		Then for every $\epsilon>0$ there exist measurable functions $A(s),B(s)$ such that for every $s\in \mathbb{R}$,
		\begin{align*}
			&X(s)=A(s) B(s),\\
			&\alpha(A(s))\leq 1+\epsilon,\\
			&\beta(B(s))\leq \gamma(X(s)).
		\end{align*}
	\end{lemma}
	
	To prove this lemma we make use of a special case of the Kuratowski--Ryll-Nardzewski selection theorem \cite{RS}:
	
	\begin{theorem} (Kuratowski--Ryll-Nardzewski selection theorem)
		Let $\emph{cl}(\mathbb{R}^n)$ be the set of all nonempty closed subsets of $\mathbb{R}^n$. Let $\psi:\mathbb{R}\rightarrow \emph{cl}(\mathbb{R}^n)$ be a multi-function such that the set $\{s\in \mathbb{R}:\psi(s)\cap U\neq \emptyset\}$ is measurable for every closed subset $U$ of $\mathbb{R}^n$. Then there is a measurable function $f:\mathbb{R}\rightarrow \mathbb{R}^n$ such that $f(s)\in \psi(s)$ for all $s\in \mathbb{R}$.
	\end{theorem}
	
	\begin{proof}[Proof of Lemma \ref{selection lemma}]
		Based on the above selection theorem we define $\psi(s)$ to be the set of all almost optimizer pairs $(A,B)$ satisfying
		\begin{align*}\label{equations}
			&X(s)=A B,\\
			&\alpha(A)\leq 1+\epsilon,\\
			&\beta(B)\leq \gamma(X(s)).
		\end{align*}
		We need to show that $\psi(\cdot)$ satisfies the hypotheses of the above selection theorem. It is straightforward to show that $\psi(s)$ is non-empty and closed for every $s$. Thus, we need to verify that $\{s\in \mathbb{R}:\psi(s)\cap U\neq \emptyset\}$ is measurable. To this end, consider the map
		\begin{equation*}
			\Psi:(s,A,B)\mapsto \Big(\gamma(X(s)-A B),\,\alpha(A),\,\gamma(X(s))-\beta(B)\Big).
		\end{equation*}
		Since $\Psi$ is continuous, the set
		\begin{equation*}
			\mathcal E:=\Psi^{-1}\Big(\{0\}\cross [0,1+\epsilon]\cross [0,+\infty)\Big)
		\end{equation*}
		is closed.
		Now, if $U$ is a closed set of pairs $(A,B)$, then $\mathcal{E}\cap (\mathbb{R}\cross U)$ is closed. Next, observe that 
		\begin{equation*}
			\{s\in \mathbb{R}:\psi(s)\cap U\neq \emptyset\}=\pi_1(\mathcal{E}\cap (\mathbb{R}\cross U)),
		\end{equation*}
		where $\pi_1$ denotes the projection map onto the first component. Therefore, to see the measurability of $\{s\in \mathbb{R}:\psi(s)\cap U\neq \emptyset\}$, it suffices to prove that $\pi_1(\mathcal{E}\cap (\mathbb{R}\cross U))$ is measurable.
		
		Note that any Euclidean space can be written as a countable union of compact sets. Then, since the intersection of a closed and a compact set is compact, the closed set $\mathcal{E}\cap (\mathbb{R}\cross U)$ can be written as
		\begin{equation*}
			\mathcal{E}\cap (\mathbb{R}\cross U)=\bigcup_{i=1}^\infty K_i
		\end{equation*}
		where $K_i$'s are compact. On the other hand, since a continuous image of a compact set is compact, $\pi(K_i)$'s are compact too. This implies that
		\begin{equation*}
			\pi_1(\mathcal{E}\cap (\mathbb{R}\cross U))=\bigcup_{i=1}^\infty \pi_1(K_i)
		\end{equation*}
		is a countable union of compact sets, and is therefore measurable.
	\end{proof}
	
	\subsection{A subharmonicity lemma}
	
	Let $\mathbb{D}$ be the unit disk in the complex plane and $\partial \mathbb{D}$ be its boundary.   It is well-known that having a sufficiently nice function $f$ on $\partial \mathbb{D}$, there is a unique \emph{harmonic} function on $\mathbb D$ that matches $f$ on the boundary. This harmonic function can be represented in terms of  the \emph{Poisson kernel} for $\mathbb{D}$, that is given by
	$$
	P(z,e^{it}):=\Re\left(\frac{e^{it}+z}{e^{it}-z}\right)=\frac{1-|z|^2}{|e^{it}-z|^2},
	$$
	for $z=re^{i\theta}\in \mathbb{D},e^{it}\in \partial \mathbb{D}$.
	In particular, if $f\in L_1(\partial \mathbb{D})$ and
	\begin{equation}\label{eq:Poisson-D}
		F(z):=\frac{1}{2\pi}\int_{-\pi}^{\pi} f(e^{it}) P(z,e^{it})\dd t,
	\end{equation}
	then,
	$$
	\lim_{r\rightarrow 1} F(re^{i\theta})=f(e^{i\theta}), \qquad \text{a.e. on} \ \partial \mathbb{D}.
	$$
	Recall that $\mathbb{S}:=\{z\in\mathbb{C}:0\leq\text{Re}(z)\leq1\}$. The Poisson kernel for $\mathbb{S}$ may be obtained by carrying $P$ over to $\mathbb{S}$ using the conformal map $\phi:\mathbb{D}\rightarrow \mathbb{S}$ 
	defined by\footnote{See, e.g., \cite{Rudin,SS} for more information on the Poisson kernel as well as conformal maps.}
	\begin{equation}\label{eq:conformal-map}
		\phi(z)=\frac{1}{i\pi}\log(i\frac{1+z}{1-z}).
	\end{equation}
	Indeed, for any  holomorphic function $f:\mathbb{S}\rightarrow \mathbb{C}$, the function $f\circ \phi:\mathbb{D}\to \C$ is also holomorphic. Thus, using~\eqref{eq:Poisson-D} we have  
	$$
	f(z)=\frac{1}{2\pi}\int_{-\pi}^{\pi} f\circ \phi(e^{it}) \cdot P\big(\phi^{-1}(z),e^{it}\big)\dd t.
	$$
	We note that $\phi$ maps the boundary of $\mathbb{D}$ onto the boundary of $\mathbb{S}$. Then, letting $\phi(e^{it})=b+is$ with $b\in \{0,1\}$ and $s\in \R$, and rewriting the integral in terms of $s$, we find that 
	\begin{equation*}
		f(z)=\sum_{b=0}^1 \int_{-\infty}^{+\infty} f(b+is)\cdot\frac{P\big(\phi^{-1}(z),\phi^{-1}(b+is)\big)}{2\cosh(\pi s)}\dd s.
	\end{equation*}
	Consequently, the Poisson kernel of $\mathbb{S}$ is given by the functions $P_b(x+iy,s)$ for $b=0,1$, where
	\begin{align*}
		P_b(x+iy,s) &:=\frac{P\big(\phi^{-1}(z),\phi^{-1}(b+is)\big)}{2\cosh(\pi s)}\\
		&= \frac{\sin(\pi x)}{2\big(\cosh(\pi(y-s))-\cos(\pi(x-b))\big)}.
	\end{align*}
	We note that $P_b(x+iy,s)\geq0$ and it can be verified that
	\begin{equation*}
		\int_{-\infty}^{+\infty} P_0(x+iy,s)\dd s=1-x,
		\qquad \quad
		\int_{-\infty}^{+\infty} P_1(x+iy,s)\dd s=x.
	\end{equation*}

	\begin{lemma}\label{lem:sub-harmonicity}
		Let $(\cX_0, \cX_1)$ be a compatible couple and let $\cX_\theta= [\cX_0, \cX_1]_\theta$ for $0<\theta<1$. 
		\begin{itemize}
			\item[{\rm (i)}] Let $f(z)\in \mathscr{F}(\mathcal{X}_0,\mathcal{X}_1)$. Then we have
			\begin{equation*}
				\log\|f(\theta)\|_{\mathcal{X}_\theta}\leq \sum_{b=0}^1 \int_{-\infty}^{+\infty} P_b(\theta, s)\log\|f(b+is)\|_{\mathcal{X}_b}\dd s.
			\end{equation*}
			
			\item[{\rm (ii)}] If $\cX_0, \cX_1$ are finite-dimensional (or more generally, reflexive), then the same inequality holds whenever $f(z)$ is bounded on $\mathbb{S}$ and holomorphic in its interior (though not necessarily continuous up to the boundary).
		\end{itemize}
	\end{lemma}
	
	The first part of this lemma is proven in~\cite[Lemma 4.3.2]{BerghLofstrom}. Here, we prove the second part for the sake of completeness. 
	
	\begin{proof}[Proof of {\rm (ii)}]
		Let $\cX^*_\theta$ be the Banach dual of $\cX_\theta$. Then, there is $\ell\in \cX^*_\theta$ such that 
		$$\|\ell\|_{\cX^*_\theta}=1, \qquad \|f(\theta)\|_{\cX_\theta} = |\langle \ell, f(\theta)\rangle|.$$
		Since $\cX^*_\theta =[\cX^*_0, \cX^*_1]_\theta$, for any $\epsilon>0$ there is $g(z)$ such that $g(\theta)=\ell$ and 
		$$\max \big\{\sup_t \|g(it)\|_{\cX^*_0},  \sup_t \|g(1+it)\|_{\cX^*_1}\big\}\leq \|\ell\|_{\cX^*_\theta}+\epsilon=1+\epsilon.$$
		Now define $\psi(z):= \langle g(z), f(z)\rangle$. We note that $\psi(z)$ is holomorphic and bounded. Then, by the subharmonicity of holomorphic functions we have
		\begin{align*}
			\log\|f(\theta)\|_{\cX_\theta}& = \log|\psi(\theta)|\\
			& \leq \sum_{b=0}^1 \int_{-\infty}^{+\infty} P_b(\theta, s)\log|\psi(b+is)|\dd s\\
			& \leq \sum_{b=0}^1 \int_{-\infty}^{+\infty} P_b(\theta, s)\log \big(\|g(b+is)\|_{\cX_b^*}\cdot \|f(b+is)\|_{\cX_b}\big)\dd s\\
			& \leq \sum_{b=0}^1 \int_{-\infty}^{+\infty} P_b(\theta, s)\log \big((1+\epsilon) \|f(b+is)\|_{\cX_b}\big)\dd s\\
			& \leq \log(1+\epsilon)+ \sum_{b=0}^1 \int_{-\infty}^{+\infty} P_b(\theta, s)\log  \|f(b+is)\|_{\cX_b}\dd s.
		\end{align*}
		The desired inequality follows once we note that $\epsilon>0$ is arbitrary.
	\end{proof}


	\section{Factorization of operator-valued functions}\label{sec:Wiener--Masani}
	
	This section is dedicated to the Wiener--Masani theorem on factorization of operator-valued functions \cite{Wiener,WienerMasani,Devinatz}. This theorem is used in the proof of Theorem~\ref{KouPis}, and is one of the main technical tools in the development of the theory of operator-valued Schatten spaces. Thus, we devote this section to the statement and proof of this theorem. 
	We follow the approach of~\cite{Helson} for the proof, that is based on characterization of \emph{invariant subspaces} of the shift operator. To this end, we first need a number of ingredients.

	\subsection{Hardy spaces}
	Following the notation of Section~\ref{sec:interpolation}, let $\D$ be the unit disk in the complex plane and let $\partial \D$ be its boundary, equipped with the normalized Lebesgue measure. Then, $L_p= L_p(\partial \D)$ denotes the Banach space of all measurable functions $f:\partial{\D}\rightarrow \mathbb{C}$ such that the norm
	$$
	\|f\|_p:=
	\begin{cases}
		\left(\frac{1}{2\pi}\int_{-\pi}^{\pi} |f(e^{it})|^p \dd t\right)^\frac1{p},  ~~~~~\qquad\qquad 1 \leq p<\infty, \\[.1in]
		\esssup |f|, ~~~~~\qquad\qquad\qquad\qquad \ \ \ p=\infty,
	\end{cases}
	$$
	is finite. We note that by H\"older's inequality, we have $L_q\subseteq L_p$ for $p<q$.

	\begin{definition}
		The \emph{Hardy space} $H_p$ is the closed subspace of $L_p$ consisting of all functions $f\in L_p$ for which there is a function $F$ holomorphic in the interior of $\D$ whose \emph{radial limit} coincides with $f$ almost everywhere.\footnote{The radial limit of $F$ at $e^{it}\in \partial \D$ equals $\lim_{r\to 1^-} F(re^{it})$.}
	\end{definition}
	There is another description of $H_p$ in terms Fourier coefficients and Fourier series that goes as follows. For each function $f\in L_1$, let
	\begin{equation*}
		a_n=\frac{1}{2\pi}\int_{-\pi}^{\pi} f(e^{it})e^{-nit}\dd t,
	\end{equation*}
	and
	\begin{equation*}
		f(e^{it})\sim\sum_{n=-\infty}^{\infty}a_n e^{nit},
	\end{equation*}
	be the Fourier coefficients and the Fourier series of $f$, respectively. Then, $H_p$ is the subset of $L_p$ consisting of all functions $f$ whose Fourier coefficients $a_n$ vanish for $n<0$.
	
	Define the operator $\fS$ by $(\fS f)(e^{it})=e^{it}f(e^{it})$. Consequently, for $n\geq 1$ we have 
	$(\fS^n f)(e^{it}) = e^{nit}f(e^{it})$. 
	The operator $\fS$ is called the \emph{shift operator} since it shifts the Fourier coefficients.

	\begin{definition}\label{def:outer-scalar}
		A function $f\in H_2$ is said to be \emph{outer} if the linear span of the set $\{\fS^n f:n\geq0\}$ is dense in $H_2$. 
	\end{definition}
	
	It is well-known that every outer function $f$ can be written in the form
	\begin{equation*}
		f(z)=\gamma\exp\left(\frac{1}{2\pi}\int_{-\pi}^{\pi}\frac{e^{it}+z}{e^{it}-z}k(t)\dd t\right),
	\end{equation*}
	where $\gamma$ is a complex number of modulus $1$, and $k(t)=\log|f(e^{it})|$ almost everywhere; see, e.g.,~\cite{Markus,Rudin} for more information on outer functions.

	\subsection{Vector-valued Hardy spaces}
	
	Let $\mathcal{H}$ be a separable (usually finite-dimensional) Hilbert space. An $\mathcal{H}$-valued function $F:\partial{\D}\to \cH$ is said to be measurable if for every $v\in\mathcal{H}$ the complex function $e^{it}\mapsto \langle v,F(e^{it})\rangle  $ is measurable.
	Abusing the notation, let $L_p(\mathcal{H})=L_p(\partial \D, \mathcal{H})$ be the Banach space of all measurable functions $F:\partial{\D}\rightarrow \mathcal{H}$ such that the norm
	$$
	\|F\|_p:=
	\begin{cases}
		\left(\frac{1}{2\pi}\int_{-\pi}^{\pi} \|F(e^{it})\|^p \dd t\right)^\frac1{p},  ~~~~~\qquad\qquad 1 \leq p<\infty, \\[.1in]
		\esssup \|F(e^{it})\|, ~~~~~\qquad\qquad\qquad\qquad \ \ \ p=\infty
	\end{cases}
	$$
	is finite.\footnote{$\|F(e^{it})\|$ is the norm of vector $F(e^{it})\in \cH$.} It is worth noting that the above integral makes sense since the function $e^{it}\mapsto\|F(e^{it})\|$, as the supremum of measurable functions, is itself measurable. The case $p=2$ is of particular interest since $L_2(\mathcal{H})$ is a Hilbert space under the inner product
	\begin{equation*}
		\langle F,G\rangle  =\frac{1}{2\pi}\int_{-\pi}^{\pi}\big\langle F(e^{it}),G(e^{it})\big\rangle  \dd t.
	\end{equation*}
	
	Let $\{e_n:n=1,2,\cdots\}$ is an orthonormal basis for $\mathcal{H}$, and for $F\in L_2(\cH)$ define the \emph{coordinate functions} by
	\begin{equation}\label{eq:coordinate-functions}
		f_n(e^{it}):=\big\langle e_n,F(e^{it})\big\rangle.
	\end{equation}
	Then, we have
	\begin{equation*}
		\|F\|_2^2=\sum_n \|f_n\|_2^2,
	\end{equation*}
	where $\|f_n\|_2$ is the norm of $f_n$ in $L_2=L_2(\partial \D, \C)$.

	For any $F\in L_p(\mathcal{H})$, one may define the \emph{Fourier expansion} of $F$ as 
	\begin{equation}\label{eq:F-Fourier-expansion}
		F(e^{it})\sim \sum_{n=-\infty}^{\infty} \hat{a}_n\cdot e^{nit},
	\end{equation}
	where $\hat{a}_n\in\mathcal{H}$ is given by
	$$\hat a_n=\frac{1}{2\pi}\int_{-\pi}^{\pi} F(e^{it})e^{-nit}\dd t.$$
	We note that letting $\cH=\C$, we obtain the usual $L_p=L_p(\partial \D, \C)$ space with the usual Fourier transform.

	\begin{definition}
		Let $\mathcal{H}$ be a (separable) Hilbert space and $1\leq p\leq\infty$. The \emph{vector-valud Hardy space} $H_p(\mathcal{H})$ is defined to be the closed subspace of $L_p(\mathcal{H})$ consisting of all functions $F$ for which $\hat{a}_n=0$, for all $n<0$.
	\end{definition}
	
	Similar to the usual complex-valued Hardy space, $H_p(\mathcal{H})$ can also be characterized in terms of holomorphic functions. We say that a vector-valued function $\phi: \D\rightarrow\mathcal{H}$ is holomorphic if for every $v\in \cH$, the function $z\mapsto\langle v,\phi(z)\rangle $ is holomorphic. Then, $H_p(\cH)$ is the space of functions $F\in L_p(\mathcal{H})$ for which there is a holomorphic function $\phi: \D\rightarrow\mathcal{H}$ whose \emph{radial limit} coincides with $F$ almost everywhere, i.e.,
	$$F(e^{it}) = \lim_{r\to 1^+} \phi(re^{it}), \qquad \text{a.e.}$$

	\subsection{Invariant subspaces}
	
	The shift operator $\fS$ can also be defined on the space of vector-valued functions: $(\fS F)(e^{it})=e^{it}F(e^{it})$.

	\begin{definition}
		A closed subspace $\mathcal{M}$ of $L_2(\mathcal{H})$ is said to be \emph{invariant} if $\fS(\mathcal{M})\subseteq\mathcal{M}$, where $\fS(\mathcal{M}):=\{\fS F:\, F\in\mathcal{M}\}$. We say that $\mathcal{M}$ is \emph{doubly invariant} if $\fS(\mathcal{M})\subseteq\mathcal{M}$ and $\fS^{-1}(\mathcal{M})\subseteq\mathcal{M}$. Also, $\mathcal{M}$ is said to be \emph{simply invariant} if it is invariant but not doubly invariant. 
	\end{definition}

	We may also consider operator-valued (matrix-valued) functions. We say that such a function $U(e^{it})$ is measurable if the function $e^{it}\mapsto\langle v,U(e^{it}) w\rangle  $ is measurable for all vectors $v, w$.

	\begin{lemma}\label{lem:char-invariant}
		Let $\mathcal{M}$ be a simply invariant subspace of $L_2(\mathcal{H})$ that contains no nontrivial doubly invariant subspace. Then there exist an auxiliary Hilbert space $\mathcal{K}$ and a measurable operator-valued function $U$ such that
		\begin{itemize}	
			\item[{\rm (i)}] $U(e^{it}):\mathcal{K}\rightarrow\mathcal{H}$ is an isometry for any $e^{it}\in\partial{\D}$.
			
			\item[{\rm (ii)}] $\mathcal{M}$ consists of functions $e^{it}\mapsto U(e^{it})F(e^{it})$ for $F\in H_2(\mathcal{K})$.
		\end{itemize}	
	\end{lemma}
	
	We note that vector-valued Hardy spaces $H_2(\mathcal{K})$ are simply invariant subspaces and do not contain any nontrivial doubly invariant subspace. The above lemma says that all such subspaces are essentially vector-valued Hardy spaces.

	\begin{proof}
		For each $n$, define $\mathcal{M}_n:=\fS^n(\mathcal{M})$. As $\mathcal{M}$ is invariant, we have
		\begin{equation*}
			\mathcal{M}=\mathcal{M}_0\supseteq\mathcal{M}_1\supseteq\mathcal{M}_2\supseteq\cdots.
		\end{equation*}
		We note that 
		$$\mathcal{M}_\infty=\bigcap\limits_{n\geq0}\mathcal{M}_n$$ 
		is doubly invariant, so by assumption $\mathcal{M}_\infty=0$. 
		
		For any $n\geq0$, let $\mathcal{N}_n$ be the orthogonal complement of $\mathcal{M}_{n+1}$ in $\mathcal{M}_n$ (with respect to the inner product of $L_2(\mathcal{H})$). Thus, we have $\mathcal{M}_n=\mathcal{N}_n\oplus\mathcal{M}_{n+1}$. Moreover, since $\mathcal{M}_\infty=0$, we have
		\begin{equation}\label{eq:M-N-direct-sum}
			\mathcal{M}=\mathcal{N}_0\oplus\mathcal{N}_1\oplus\mathcal{N}_2\oplus\cdots.
		\end{equation}
		
		Let $\mathcal B_0=\{E_1,E_2,\cdots\}$ be an orthonormal basis for $\mathcal{N}_0$, and note that $\mathcal B_n:=\fS^n(\mathcal B_0)=\{\fS^n E_k:k\geq1\}$, for $n\geq1$, is a subset of $\mathcal{M}_n\subseteq\mathcal{M}_1$. Thus, $\mathcal B_0$ is orthogonal to $\mathcal B_n$ for all $n\geq1$. This means that for every $j,k$ and $n\geq1$ we have
		\begin{equation*}
			\langle E_j,\fS^n E_k\rangle  =\frac{1}{2\pi}\int_{-\pi}^{\pi} e^{nit}\big\langle E_j(e^{it}),E_k(e^{it})\big\rangle  \dd t=0.
		\end{equation*}
		Moreover, as $\mathcal B_0$ is orthonormal, the same equation holds when $n=0$ and $j\neq k$. More generally, we observe that
		\begin{equation*}
			\langle \fS^n E_j,\fS^m E_k\rangle  =\frac{1}{2\pi}\int_{-\pi}^{\pi} e^{(m-n)it}\big\langle E_j(e^{it}),E_k(e^{it})\big\rangle  dt=\delta_{jk}\cdot\delta_{mn}.
		\end{equation*}
		This equation shows that all the nontrivial Fourier coefficients of the functions $e^{it}\mapsto\big\langle E_j(e^{it}),E_k(e^{it})\big\rangle  $ are zero, and hence it is almost everywhere identical to the constant function $\delta_{jk}$.
		
		We claim that for each $n$, the orthonormal set $\mathcal B_n$ is a basis for $\mathcal{N}_n$. To this end, we need to show that $\mathcal B_n\subseteq\mathcal{N}_n=\spn \mathcal B_n$. Letting $F\in \mathcal{M}_{n+1}$, by definition there is $G\in\mathcal{M}$ such that $F=\fS^{n+1} G$. On the other hand, for $\fS^n E_j\in \mathcal B_n$ we have
		\begin{equation*}
			\langle \fS^n E_j,F\rangle  =\langle \fS^n E_j,\fS^{n+1}G\rangle  =\langle E_j,\fS G\rangle  =0,
		\end{equation*}
		where we used the fact that $\fS G\in\mathcal{M}_1$ and that $E_j\in \mathcal{N}_0$. Thus, $\fS^nE_j\in \mathcal{M}_n$ and is orthogonal to $\mathcal{M}_{n+1}$, which means that $\fS^n E_j\in \mathcal{N}_n$. Therefore, $\mathcal B_n\subseteq \mathcal N_n$.
		
		Next, let $F\in \mathcal{M}_n$. Then, there is $G\in \mathcal{M}$ such that $F=\fS^n G$. Since $G\in \mathcal{M}=\mathcal{N}_0\oplus \mathcal{M}_1$, there are $G_0\in \mathcal{N}_0$ and $G_1\in \mathcal{M}$ such that $G=G_0+\fS G_1$. Therefore
		\begin{equation*}
			F=\fS^n G_0+\fS^{n+1}G_1.
		\end{equation*}
		By definition we have $\fS^{n+1}G_1\in \mathcal{M}_{n+1}$. Also, since $G_0\in \mathcal{N}_0=\spn \mathcal B_0$, we have $\fS^n G_0\in \spn \mathcal B_n$. As a result, $F\in \spn \mathcal B_n\oplus \mathcal{M}_{n+1}$ for any $F\in \mathcal{M}_n$. Thus, using $\mathcal B_n\subseteq \mathcal N_n$ proven above, and the definition of $\mathcal N_n$, we find that $\mathcal N_n=\spn \mathcal B_n$.
		
		Putting~\eqref{eq:M-N-direct-sum} and $\mathcal N_n=\spn \mathcal B_n$ together, we find that every $F\in \mathcal{M}$ can be expanded as
		\begin{equation}\label{eq:F-expansion-fS}
			F=\sum_j \sum_{n=0}^\infty a_{nj}\fS^n E_j, \ \ \  \quad \sum_{n,j} |a_{nj}|^2<\infty.
		\end{equation}
		Now, let $\mathcal{K}$ be a Hilbert space isomorphic to $\mathcal{N}_0$ with orthonormal basis $\{e_1,e_2,\cdots\}$. Every function in $L_2(\mathcal{K})$ can be written in the form $\sum f_j e_j$ with $f_j\in L_2$, $j=1, 2,\dots$, being coordinate functions (see \eqref{eq:coordinate-functions}).
		Define $\hat{U}:L_2(\mathcal{K})\rightarrow L_2(\mathcal{H})$ by
		\begin{equation*}
			\hat{U}\Big(\sum f_j e_j\Big)=\sum f_j E_j,
		\end{equation*}
		meaning that $\hat{U}\Big(\sum f_j e_j\Big)(e^{it})=\sum f_j(e^{it}) E_j(e^{it})$.
		Then, $\hat{U}$ is a well-defined linear operator. Moreover, since $\langle E_j(e^{it}),E_k(e^{it})\rangle  =\delta_{jk}$, almost everywhere, $\hat{U}$ is an isometry.

		For every $e^{it}$ we define a linear map $U(e^{it}):\mathcal{K}\rightarrow\mathcal{H}$ as follows. For every $v\in \cK$ let $\hat v\in L_2(\cK)$ be the constant function $\hat{v}(e^{it})=v$. Then, define 
		\begin{equation*}
			U(e^{it})v:=(\hat{U}\hat{v})(e^{it}).
		\end{equation*}
		It is easy to verify that $e^{it}\mapsto U(e^{it})$ is measurable.
		We claim that $U(e^{it})$ is an isometry for almost every $t$. To show this, fix $v, w\in \mathcal{K}$. Using the fact that $\hat{U}$ is an isometry, for every $n$ we obtain
		\begin{align}
			\big\langle \hat{U}\hat{v},\hat{U}(\fS^n\hat{w})\big\rangle  
			&=\langle \hat{v},\fS^n\hat{w}\rangle \nonumber \\
			&=\frac{1}{2\pi}\int_{-\pi}^{\pi} e^{nit}\langle v, w\rangle  \dd t \nonumber\\
			&=\frac{1}{2\pi}\langle v, w\rangle  \int_{-\pi}^{\pi} e^{nit}\dd t\nonumber\\
			&=\delta_{n0}\langle v, w\rangle  .
			\label{eq:inner-prod-Uv-Uw-0}
		\end{align}
		On the other hand, the definition of $\hat{U}$ implies that $\hat{U}(\fS^n \hat{w})=\fS^n\hat{U}(\hat{w})$. Therefore,
		\begin{align*}
			\langle \hat{U}\hat{v},\hat{U}(\fS^n\hat{w})\rangle  
			&=\langle \hat{U}\hat{v}, \fS^n\hat{U}(\hat{w})\rangle  \\
			&=\frac{1}{2\pi}\int_{-\pi}^{\pi} e^{nit}\langle \hat{U}\hat{v}(e^{it}),\hat{U}\hat{w}(e^{it})\rangle  \dd t.
		\end{align*}
		Consequently, comparing to the previous equation, all the nontrivial Fourier coefficients of the function $e^{it}\mapsto\langle \hat{U}\hat{v}(e^{it}),\hat{U}\hat{w}(e^{it})\rangle =\langle U(e^{it})v, U(e^{it})w\rangle   $ vanish. This means that this function is constant almost everywhere. Thus, we have, almost everywhere,
		\begin{align*}
			\langle U(e^{it})v, U(e^{it})w\rangle  
			&=\frac{1}{2\pi}\int_{-\pi}^{\pi}\langle U(e^{is})v, U(e^{is})w\rangle  \dd s\\
			&=\frac{1}{2\pi}\int_{-\pi}^{\pi}\langle \hat{U}\hat{v}(e^{is}), \hat{U}\hat{w}(e^{is})\rangle  \dd s\\
			&=\langle \hat{U}\hat{v}, \hat{U}\hat{w}\rangle  \\
			&=\langle v, w\rangle,  
		\end{align*}
		where in the last line we use~\eqref{eq:inner-prod-Uv-Uw-0}. This means that $U(e^{it})$ is an isometry almost everywhere (hence one can modify $U(e^{it})$ on a set of measure zero so that $U(e^{it})$ is an isometry for any $e^{it}$).

		We note that $H_2(\cK) = \big\{\sum f_je_j:\, f_j\in H_2, \|f_j\|_2^2<\infty\big\} $, and every $f_j\in H_2$ can be written  as  $f_j(e^{it})=\sum_{n\geq 0} a_{nj}e^{nit}$. Therefore,~\eqref{eq:F-expansion-fS} yields 
		\begin{align*}
			\hat{U}H_2(\mathcal{K})
			&=\Big\{\sum_j f_j E_j:f_j\in H_2(\C),\sum \|f_j\|_2^2<\infty\Big\}\\
			&=\Big\{\sum_j \sum_{n\geq0}a_{nj}\fS^nE_j:\sum |a_{nj}|^2<\infty\Big\}\\
			&=\mathcal{M}.
		\end{align*}
		Finally, from the definition of $U$ it follows that $U(e^{it})e_j=(\hat{U}e_j)(e^{it})=E_j(e^{it})$. Hence, for every $F=\sum_j f_je_j\in H_2(\mathcal{K})$ we obtain
		\begin{align*}
			(\hat{U}F)(e^{it})
			&=\Big(\sum f_j E_j\Big)(e^{it})\\
			&=\sum f_j(e^{it}) E_j(e^{it})\\
			&=\sum f_j(e^{it}) U(e^{it})e_j\\
			&=U(e^{it})\Big(\sum f_j(e^{it})e_j\Big)\\
			&=U(e^{it})F(e^{it}).
		\end{align*}
		That is, $\mathcal{M}$ consists of functions of the form $e^{it}\mapsto U(e^{it})F(e^{it})$ for $F\in H_2(\mathcal{K})$.
	\end{proof}
	
	\begin{remark}
		In the above lemma, if $\mathcal{H}$ is finite-dimensional, then so is $\mathcal{K}$ since $U(e^{it}):\mathcal{K}\rightarrow\mathcal{H}$ is an isometry almost everywhere.
	\end{remark}

	\subsection{Wiener--Masani factorization theorem}
	
	In the following we assume that $\cH$ is finite-dimensional with orthonormal basis $\{e_1, \dots, e_d\}$.  For every $e^{it}\in \partial \D$ we assume that  $T(e^{it}):\mathcal{H}\rightarrow\mathcal{H}$ is a linear map such that $e^{it}\mapsto T(e^{it})$ is an operator-valued measurable function. We further assume that entries of $T$ (in the basis $\{e_1, \dots, e_d\}$) are integrable, i.e., belong to $L_1$.
	
	Suppose that $T(e^{it})$ is positive, and uniformly bounded almost everywhere, i.e., there exists a constant $c>0$ such that 
	$$0\leq T(e^{it})\leq cI, \qquad \text{ almost everywhere}.$$
	Let $Q(e^{it})=\sqrt{T(e^{it})}$. We note that with the above assumptions $Q$ can be written as the limit of measurable functions (consider the truncated Taylor expansion of $x\mapsto\sqrt{x}$ at point $c$), and is measurable itself.
	We also note that $Q$ can be considered as an element of $B(L_2(\cH))$.  
	Indeed, for every $F\in L_2(\mathcal{H})$ we have $(QF)(e^{it})=Q(e^{it})F(e^{it})$ and
	\begin{align*}
		\|QF\|_2^2
		&=\frac{1}{2\pi}\int_{-\pi}^{\pi}\big\langle Q(e^{it})F(e^{it}),Q(e^{it})F(e^{it})\big\rangle \dd t\\
		&=\frac{1}{2\pi}\int_{-\pi}^{\pi}\langle F(e^{it}),T(e^{it})F(e^{it})\rangle  \dd t\\
		&\leq\frac{1}{2\pi}\int_{-\pi}^{\pi}c\langle F(e^{it}),F(e^{it})\rangle  \dd t\\
		&=c\|F\|_2^2.
	\end{align*}
	
	Define
	\begin{equation*}
		\mathcal{M}(T)=\overline{\spn\{\fS^n (Qe_j):n\geq0,\, j=1,\dots, d\}},
	\end{equation*}
	where $Qe_j\in L_2(\cH)$ is given by $(Qe_j)(e^{it}) = Q(e^{it})e_j$. Observe that by definition, $\mathcal{M}(T)$ is a closed and invariant subspace of $L_2(\mathcal{H})$. In fact, $\mathcal{M}(T)$ is the smallest invariant subspace of $L_2(\mathcal{H})$ generated by columns of $Q$.

	\begin{lemma}\label{invariant}
		Let $T(e^{it})$ be a measurable operator-valued function as above, and $\lambda(e^{it})$ be a positive measurable function such that
		\begin{equation}\label{eq:assumption-lambda}
			\log\lambda(e^{it})\in L_1,
		\end{equation}
		and
		\begin{equation*}
			0<\lambda(e^{it})I\leq T(e^{it})\leq cI, \qquad \text{a.e. on $\partial{\D}$},
		\end{equation*}
		for some $c>0$. Then $\mathcal{M}(T)$ does not contain any nontrivial doubly invariant subspace.

	\end{lemma}

	\begin{proof}
		Observe that by~\eqref{eq:assumption-lambda} there is $\psi\in H_2$ such that $|\psi(e^{it})|=\sqrt{\lambda(e^{it})}$. To construct such a function we use the Poisson kernel of $\D$ to define
		$$g(z) := \frac{1}{2\pi} \int_{-\pi}^\pi  \Re \frac{e^{it}+z}{e^{it}-z} \log \sqrt{\lambda(e^{it})} \dd t.$$
		We note that $g(e^{it})=\log \sqrt{\lambda(e^{it})}$ almost everywhere, and that $g(z)$ is harmonic. Then, considering its harmonic conjugate, there is a holomorphic function $\hat g(z)$ with $\Re \hat g(z) = g(z)$.  Now, letting $\psi(z) = e^{\hat g(z)}$, we have almost everywhere
		$$|\psi(e^{it})| = |e^{\hat g(it)}| = e^{\Re \hat g(it)} = e^{g(it)} = \sqrt{\lambda(e^{it})}.$$ 
		
		We say that $P\in H_2(\cH)$ is a \emph{polynomial} if its Fourier expansion~\eqref{eq:F-Fourier-expansion} consists of finitely many $n\geq0$ terms, i.e., $P$ is a polynomial in variable $e^{it}$ with coefficients in $\cH$. For such a polynomial $P\in H_2(\cH)$ we have $\psi P\in H_2(\cH)$ where $(\psi P)(z) = \psi(z)P(z)$. 
		Define
		\begin{equation*}
			\mathcal{N}=\overline{\spn\{\psi P:\text{$P$ polynomial}\}}.
		\end{equation*}
		Observe that $\mathcal{N}\subseteq H_2(\mathcal{H})$ is an invariant subspace. 
		
		Next, consider the subspace $\mathcal W=\spn\{Q P:\, \text{$P$  polynomial}\}\subseteq L_2(\cH)$ where as before $Q=\sqrt T$.  We note that by definition $\overline{\mathcal W}= \mathcal M(T)$.
		Define $U$ from $\mathcal W$ to $\mathcal{N}$ by $U(Q P)=\psi P$. Since $\sqrt{\lambda(e^{it})} I\leq Q(e^{it})$ almost everywhere, we have
		\begin{align*}
			\|\psi P\|^2_2 &  = \frac{1}{2\pi}\int_{-\pi}^\pi \big\|\psi(e^{it}) P(e^{it}) \big\|^2\\
			&  = \frac{1}{2\pi}\int_{-\pi}^\pi  \lambda(e^{it}) \big\| P(e^{it}) \big\|^2\\
			&  \leq \frac{1}{2\pi}\int_{-\pi}^\pi  \big\| (QP)(e^{it}) \big\|^2\\
			& = \|QP\|^2_2.
		\end{align*}
		This means that $\|U\|\leq1$, and $U$ can be extended to $U:\mathcal{M}(T)\rightarrow\mathcal{N}$ while $\|U\|\leq1$.
		
		We claim that $U$ is one-to-one. To prove this, suppose that $F\in\mathcal{M}(T)$ and $UF=0$. Since $F\in\mathcal{M}(T)$, there is a sequence of polynomials $\{P_n:\, n\geq 1\}\subset H_2(\cH)$ such that
		\begin{equation*}
			\lim_{n\rightarrow\infty}Q P_n=F \ \ \ \qquad \text{(in $L^2(\mathcal{H})$)},
		\end{equation*}
		and since $QF=0$, we have
		\begin{equation*}
			\lim_{n\rightarrow\infty}\psi P_n=0 \ \ \ \qquad \text{(in $L^2(\mathcal{H})$)}.
		\end{equation*}
		Define $R_m=\{e^{it}:\lambda(e^{it})\geq1/m\}$, and let $J_m$ be the characteristic function of $R_m$. We compute
		\begin{align*}
			0&\leq\frac{1}{2\pi}\int_{-\pi}^{\pi}\frac{1}{m} J_m(e^{it})\big\|P_n(e^{it})\big\|^2\,\dd t\\
			&\leq\frac{1}{2\pi}\int_{-\pi}^{\pi} \lambda(e^{it})J_m(e^{it})\big\|P_n(e^{it})\big\|^2\,\dd t\\
			&\leq\frac{1}{2\pi}\int_{-\pi}^{\pi} \lambda(e^{it})\big\|P_n(e^{it})\big\|^2\,\dd t\\
			&=\|\psi P_n\|_2^2.
		\end{align*}
		Therefore
		\begin{equation}\label{eq:limit-zero-all-m}
			\lim_{n\rightarrow\infty}\frac{1}{2\pi}\int_{-\pi}^{\pi} J_m(e^{it})\big\|P_n(e^{it})\big\|_2^2\,\dd t=0, \ \ \ \quad \forall m.
		\end{equation}
		On the other hand, we have
		\begin{align*}
			&\big\|F-QP_n\big\|_2^2\\
			&=\frac{1}{2\pi}\int_{-\pi}^{\pi} \big\|F(e^{it})-Q(e^{it})P_n(e^{it})\big\|^2\,\dd t\\
			&\geq\frac{1}{2\pi}\int_{-\pi}^{\pi} J_m(e^{it})\big\|F(e^{it})-Q(e^{it})P_n(e^{it})\big\|^2\,\dd t\\
			&\geq\frac{1}{2\pi}\int_{-\pi}^{\pi} J_m(e^{it})\Big(\|F(e^{it})\big\|-\|Q(e^{it})P_n(e^{it})\big\|\Big)^2\,\dd t\\
			&=\frac{1}{2\pi}\int_{-\pi}^{\pi} J_m(e^{it})\Big(\big\|F(e^{it})\big\|^2+\big\|Q(e^{it})P_n(e^{it})\big\|^2-2\big\|F(e^{it})\big\|\cdot \big\|Q(e^{it})P_n(e^{it})\big\|\Big)\dd t\\
			&\geq\frac{1}{2\pi}\int_{-\pi}^{\pi} J_m(e^{it})\big\|F(e^{it})\big\|^2\,\dd t-\frac{1}{\pi}\left(\int_{-\pi}^{\pi} \big\|F(e^{it})\big\|^2\,\dd t \int_{-\pi}^{\pi} J_m(e^{it})\big\|Q(e^{it})P_n(e^{it})\big\|^2\,\dd t\right)^{\frac{1}{2}}\\
			&\geq\frac{1}{2\pi}\int_{-\pi}^{\pi} J_m(e^{it})\big\|F(e^{it})\big\|^2\,\dd t-2\|F\|_2\cdot\frac{c}{2\pi}\int_{-\pi}^{\pi} J_m(e^{it})\big\|P_n(e^{it})\big\|^2\,\dd t,
		\end{align*}
		where the penultimate inequality follows by ignoring a non-negative term and the Cauchy-Schwarz inequality, and the last inequality follows from $Q(e^{it})\leq \sqrt cI$. Taking the limit of both sides as $n\rightarrow\infty$, using~\eqref{eq:limit-zero-all-m} and the non-negativity of the remaining term, we obtain
		\begin{equation*}
			\|J_m F\|_2^2=\frac{1}{2\pi}\int_{-\pi}^{\pi} J_m(e^{it})\big\|F(e^{it})\big\|^2\,\dd t=0, \ \ \ \quad \forall m.
		\end{equation*}
		Now, the monotone convergence theorem implies $\|F\|_2^2=\lim\limits_{m\rightarrow\infty} \|J_mF\|_2^2=0$.
		This means that $F=0$. Thus, $U:\mathcal{M}(T)\rightarrow\mathcal{N}$ is one-to-one, as claimed.
		
		Finally, suppose that $\mathcal{M}'\subseteq\mathcal{M}(T)$ is a nontrivial doubly invariant subspace. Then $\overline{U\mathcal{M}'}\subseteq\mathcal{N}\subseteq H_2(\mathcal{H})$ is a doubly invariant subspace. But $H_2(\mathcal{H})$ does not contain nontrivial doubly invariant subspaces. This implies that $\overline{U\mathcal{M}'}=0$. Hence, $\mathcal{M}'=0$ since $U$ is one-to-one.
	\end{proof}

	Similar to the scalar case (see Definition~\ref{def:outer-scalar}), an operator-valued function $A(z)$ is called outer if $\{A(z)P(z):\, P(z)\in H_2(\cH) \text{ polynomial}\}$ is dense in $H_2(\cH)$. Equivalently, $A(z)$ is outer if 
	$$\overline{\spn\{\fS^n (Ae_j):\, n\geq0,\, j=1,\dots, d\}} = H_2(\cH).$$

	\begin{lemma}[\cite{Markus}]\label{invert1}
		Let $A(z)$ be an outer operator-valued function in $H_\infty(\mathcal{H})$. If $A(e^{it})$ is invertible almost everywhere on $\partial{\D}$, and
		\begin{equation}\label{invert2}
			\log\|A^{-1}(e^{it})\|\in L_1,
		\end{equation}
		then $A(z)$ is invertible for all $z$ in the interior of $\D$.
	\end{lemma}
	
	\begin{proof}
		Since $\|A^{-1}(e^{it})\|^{-1} \leq \|A(e^{it})\|$ almost everywhere on $\partial{\D}$ and $A(z) \in H_\infty(\mathcal{H})$, it is evident that the function $\tilde{\lambda}(e^{it}):=\|A^{-1}(e^{it})\|^{-1}$ lies in $L_\infty$. On the other hand, it follows from \eqref{invert2} that $\log\tilde{\lambda}(e^{it})\in L_1$. These facts imply that the outer function $\lambda(z)$, defined by
		$$
		\lambda(z):=\exp\left(\frac{1}{2\pi}\int_{-\pi}^{\pi}\frac{e^{it}+z}{e^{it}-z}\log\tilde{\lambda}(e^{it})\dd t\right),
		$$
		lies in $H_\infty$, and
		\begin{equation}\label{B11}
			\big\|\lambda(e^{it}) A^{-1}(e^{it})\big\|=1, \qquad \text{a.e.\ on $\partial{\D}$}.
		\end{equation}
		Let $h$ be an arbitrary but fixed vector in $\mathcal{H}$, and define $\psi_h(e^{it})=\overline{\lambda(e^{it})} A^{-1}(e^{it})^*h$. Then, $\psi_h \in L_2(\mathcal{H})$ since $\|\psi_h\|_{L_2(\mathcal{H})} \leq \|h\|_\mathcal{H}$ by \eqref{B11}.
		
		Now let $\Psi_h$ be the bounded linear functional on $L_2(\mathcal{H})$ corresponding to $e^{-it}\psi_h(e^{it}) \in L_2(\mathcal{H})$, i.e.,
		\begin{align*}
			\Psi_h(F)
			&=\big\langle e^{-it}\psi_h(e^{it}), F(e^{it}) \big\rangle_{L_2(\mathcal{H})}\\
			&=\frac{1}{2\pi}\int_{-\pi}^{\pi} e^{it}\big\langle \psi_h(e^{it}), F(e^{it}) \big\rangle_{\mathcal{H}}\dd t.
		\end{align*}
		We are going to show that $\Psi_h$ vanishes on $H_2(\mathcal{H})$. To this end, note first that the last expression above is indeed a Fourier coefficient, with negative index, of the function $\big\langle \psi_h(e^{it}), F(e^{it}) \big\rangle_{\mathcal{H}} \in H_2$. As a result, for every function $u(z) \in H_2(\mathcal{H})$ we have
		$$
		\Psi_h(F)=\frac{1}{2\pi}\int_{-\pi}^{\pi} e^{it}\big\langle \overline{\lambda(e^{it})}h , u(e^{it})\big\rangle_{\mathcal{H}}\dd t=0,
		$$
		where $F(z)=A(z)u(z)$, since $\big\langle \overline{\lambda(e^{it})}h, u(e^{it}) \big\rangle_{\mathcal{H}} \in H_2$. Thus, $\Psi_h$ vanishes on $A(z) H_2(\mathcal{H})$, which in turn implies that $\Psi_h$ vanishes on the entire $H_2(\mathcal{H})$ as by assumption $A(z)$ is outer. 
		
		Let $g$ be an arbitrary (but fixed) vector in $\mathcal{H}$, and $z \notin \D$. Then, the vector-valued function $F^{g}_z(e^{it}):=(e^{it}-z)^{-1}g$ is readily seen to be in $H_2(\mathcal{H})$. Hence, by the preceding paragraph  $\Psi_h(F^{g}_z)=0$, or,
		$$
		\frac{1}{2\pi}\int_{-\pi}^{\pi}\frac{\lambda(e^{it})\big\langle h, A^{-1}(e^{it})g \big\rangle_{\mathcal{H}}}{e^{it}-z}e^{it}\dd t=0.
		$$
		Consequently, letting $B(e^{it})=A^{-1}(e^{it})$, the function $\lambda(e^{it})\big\langle h, A^{-1}(e^{it})g \big\rangle_{\mathcal{H}}$ extends to a function $\lambda(z)\langle h, B(z)g \rangle_{\mathcal{H}}$ in $H_\infty$ via Poisson integral. Finally, since $\lambda(z)$ is outer, we have $\lambda(z)\neq 0$ for all $z$ in the interior of $\D$, and hence the function $\langle h, B(z)g \rangle_{\mathcal{H}}$ is holomorphic as well. It is easily verified that indeed $B(z)=A^{-1}(z)$ for all $z$ in the interior of $\D$.
	\end{proof}

	We can now state and prove the Wiener--Masani theorem. 
	
	\begin{theorem}[Wiener--Masani]
		Let $T(e^{it})$ be a measurable operator-valued function as above, and $\lambda(e^{it})$ be a positive measurable function such that
		\begin{equation}\label{hypo1}
			\log\lambda(e^{it})\in L_1,
		\end{equation}
		and
		\begin{equation}\label{hypo2}
			0<\lambda(e^{it})I\leq T(e^{it})\leq cI, \qquad \text{a.e. on $\partial{\D}$},
		\end{equation}
		for some $c>0$. Then there exists an outer operator-valued function $A(z)\in H_\infty(\mathcal{H})$ such that $A(z)$ is invertible for all $z$ in the interior of $\mathbb{D}$, and
		$$
		T(e^{it})=A^*(e^{it})A(e^{it}), \qquad \text{a.e. on $\partial{\D}$}.
		$$

	\end{theorem}

	\begin{proof}
		By Lemma \ref{invariant}, $\mathcal{M}(T)$ does not contain any doubly invariant subspace. Therefore, by Lemma~\ref{lem:char-invariant}, there exist a Hilbert space $\mathcal{K}$ and a function $U(e^{it}):\mathcal{K}\rightarrow\mathcal{H}$, that is an isometry almost everywhere, such that $\mathcal{M}(T)=UH_2(\mathcal{K})$.
		
		Recall that $Qe_j\in\mathcal{M}(T)$ for every basis element $e_j\in \cH$. Then, for each $j$, there is $A_j\in H_2(\mathcal{K})$ such that $Qe_j=UA_j$. Define $A(e^{it}): \cH\to \cK$ by $A(e^{it})e_j=A_j(e^{it})$. Observe that $A$ is a holomorphic operator-valued function since $A_j\in H_2(\mathcal{K})$. Moreover, using the fact that $U$ is an isometry almost everywhere, we have
		\begin{align}\label{B12}
			A^*(e^{it}) A(e^{it})
			&=A^*(e^{it}) U^*(e^{it}) U(e^{it})A(e^{it}) \notag\\
			&=Q^*(e^{it}) Q(e^{it}) \notag\\
			&=T(e^{it}), \qquad \text{a.e.\ on $\partial{\D}$}.
		\end{align}
		
		Now we show that $A(e^{it})$ is outer. Let $F\in H_2(\cK)$ be arbitrary. Then, $UF\in UH_2(\cK)= \mathcal M(T)$, so for every $n\geq 1$ there is a polynomial $P_n$ such that $\|QP_n- UF\|_2\leq 1/n$. Next, using the fact that $U$ is an isometry almost everywhere we have
		$$\|AP_n- F\|_2= \|UAP_n- UF\|_2=\|QP_n- UF\|_2\leq 1/n.$$
		This shows that $\{AP: P \text{ polynomial}\}$ is dense in $\cH_2(\cK)$ and $A$ is outer.

		It follows from \eqref{B12} and our assumptions on $T(e^{it})$ that $A(e^{it})$ is invertible almost everywhere on $\partial{\D}$, and $\|A^{-1}(e^{it})\|^2=\|T^{-1}(e^{it})\|$. On the other hand, by \eqref{hypo2} we have
		\begin{align*}
			2\log \|A^{-1}(e^{it})\|
			&=\log \|T^{-1}(e^{it})\|\\
			&\leq \log \lambda(e^{it})^{-1}, \qquad \text{a.e. on $\partial{\D}$}.
		\end{align*}
		This, together with \eqref{hypo1}, yields that $\log \|A^{-1}(e^{it})\| \in L_1$. Finally, Lemma \ref{invert1} is applied to conclude that $A(z)$ is invertible for all $z$ in the interior of $\mathbb{D}$.
	\end{proof}
	
	The Wiener--Masani theorem can also be stated for operator-valued functions defined on the boundary of the strip $\mathbb S$. To this end, having measurable operator-valued functions $T_b(s)\in B(\cH)$, for $b=0,1$,
	we may use the conformal map~\eqref{eq:conformal-map} to transform them into measurable operator-valued function on $\partial \D$. Then, using the Wiener--Masani theorem, finding an appropriate holomorphic operator-valued function $A$, and pulling it back as a function on $\mathbb S$ through the conformal map, the desired factorization of $T_b$, $b=0,1$, is derived. 
	
	\begin{cor}\label{corr:Wiener--Masani}
		Let $T_b(b+is)$, $b=0,1$, be two measurable operator-valued functions, and $\lambda_b(b+is)$ be two positive measurable functions such that\footnote{The density $\frac{1}{\cosh(\pi s)}$ on the boundary of the strip appears here by carrying the uniform measure on the unit circle via the conformal map as explained above.}
		$$
		\log\lambda_b(b+is)\in L_1\left(\frac{\dd s}{\cosh(\pi s)}\right),
		$$
		and
		\begin{equation*}
			0< \lambda_b(b+is)I\leq T_b(b+is)\leq cI,  \qquad \text{a.e.\ on $\partial{\mathbb{S}}$},
		\end{equation*}
		for some $c>0$. Then there exists a bounded holomorphic operator-valued function $A(z)$ on $\mathbb{S}$ such that $A(z)$ is invertible for all $z$ in the interior of $\mathbb{S}$, and
		$$
		T_b(b+is)=A^*(b+is)A(b+is),  \qquad \text{a.e.\ on $\partial{\mathbb{S}}$}.
		$$

	\end{cor}
	

\end{document}